\documentclass[a4paper, 12pt]{article}
\usepackage{etoolbox, amsmath, amsfonts, amsthm, graphicx, bbm, natbib, setspace, enumitem, titling, amssymb, tikz, pgfplots, natbib, mathrsfs, rotating, booktabs, caption} 
\usepackage[flushleft]{threeparttable} 
\usepgfplotslibrary{fillbetween}
\AtBeginDocument{
	\DeclareSymbolFont{AMSb}{U}{msb}{m}{n}
	\DeclareSymbolFontAlphabet{\mathbb}{AMSb}}
\usepackage[english]{babel} 
\tikzset{>=latex}
\usepackage[colorlinks, allcolors = dark-blue]{hyperref}
\usepackage{cleveref}

\usepackage[includeheadfoot]{geometry}
\usepackage[justification=centering, margin=2cm]{caption}

\theoremstyle{plain}

\newtheorem{assumption}{Assumption}
\newtheorem{theorem}{Theorem}
\newtheorem{lemma}{Lemma}
\newtheorem{corollary}{Corollary}

\theoremstyle{definition}
\newtheorem{definition}{Definition}
\newtheorem{example}{Example}
\newtheorem{remark}{Remark}
\newtheorem*{notation}{Notation}

\theoremstyle{remark}

\definecolor{light-blue}{rgb}{0.8,0.85,1}
\definecolor{dark-blue}{rgb}{0.0, 0.0, 0.55}
\definecolor{bu-red}{rgb}{0.8, 0, 0}

\newcommand\placeqed{\nobreak\enspace$\blacksquare$}

\begin{document}	
\title{Identification and Inference for Welfare Gains \\ without Unconfoundedness \thanks{I am deeply indebted to my main adviser Hiroaki Kaido for his unparalleled guidance and constant encouragement throughout this project. I am grateful to Iv\'an Fern\'andez-Val and Jean-Jacques Forneron for their invaluable suggestions and continuous support. For their helpful comments and discussions, I thank Karun Adusumilli, Fatima Aqeel, Susan Athey, Jesse Bruhn, Alessandro Casini, Mingli Chen, Shuowen Chen, David Childers, Taosong Deng, Thea How Choon, Shakeeb Khan, Michael Leung, Arthur Lewbel, Jessie Li, Jia Li, Enkhjargal Lkhagvajav, Siyi Luo, Ching-to Albert Ma, Francesca Molinari, Guillaume Pouliot, Patrick Power, Anlong Qin, Zhongjun Qu, Enrique Sentana, Vasilis Syrgkanis, Purevdorj Tuvaandorj, Guang Zhang, and seminar participants at Boston University, University of Mannheim, University of Surrey, University of Bristol, Institute of Social and Economic Research at Osaka University and participants of BU-BC Joint Workshop in Econometrics, the Econometric Society/Bocconi University World Congress 2020, and the European Economic Association Annual Congress 2020. All errors are my own and comments are welcome.}\\
}
\author{Undral Byambadalai \thanks{Stanford Graduate School of Business Email: \href{mailto:undralb@stanford.edu}{undralb@stanford.edu} 
} 
}
\date{\today}

\maketitle

\begin{abstract}
This paper studies identification and inference of the welfare gain that results from switching from one policy (such as the status quo policy) to another policy. The welfare gain is not point identified in general when data are obtained from an observational study or a randomized experiment with imperfect compliance. I characterize the sharp identified region of the welfare gain and obtain bounds under various assumptions on the unobservables with and without instrumental variables. Estimation and inference of the lower and upper bounds are conducted using orthogonalized moment conditions to deal with the presence of infinite-dimensional nuisance parameters. I illustrate the analysis by considering hypothetical policies of assigning individuals to job training programs using experimental data from the National Job Training Partnership Act Study. Monte Carlo simulations are conducted to assess the finite sample performance of the estimators.
\vspace{0in}\\	
\vspace{0in}
\vspace{0in}
\end{abstract}
\setcounter{page}{0}
\thispagestyle{empty}

\clearpage
\onehalfspacing
\section{Introduction}
The problem of choosing among alternative treatment assignment rules based on data is pervasive in economics and many other fields, including marketing and medicine. A treatment assignment rule is a mapping from individual characteristics to a treatment assignment. For instance, it can be a job training program eligibility criterion based on the applicants' years of education and annual earnings. Throughout the paper, I call the treatment assignment rule a policy, and the subject who decides the treatment assignment rule a policymaker. The policymaker can be an algorithm assigning targeted ads, a doctor deciding medical treatment, or a school principal deciding which students take classes in person during a pandemic. As individuals with different characteristics might respond differently to a given policy, policymakers aim to choose a policy that generates the highest overall outcome or welfare.

Most previous work on treatment assignment in econometrics focused on estimating the optimal policy using data from a randomized experiment. I contribute to this literature by focusing on the identification and inference of the \emph{welfare gain} using data from an observational study or a randomized experiment with imperfect compliance. The assumption called \emph{unconfoundedness} might fail to hold for such datasets.\footnote{ The assumption of unconfoundedness is also known as \emph{selection on observables} and assumes that treatment is independent of potential outcomes conditional on observable characteristics.} By relaxing the unconfoundedness assumption, my framework accommodates many interesting and empirically relevant cases, including the use of instrumental variables to identify the effect of a treatment. The advantage of focusing on welfare gain is to provide policymakers with the ability to be more transparent when choosing among alternative policies. Policymakers may want to know how much the welfare gain or loss is in addition to the welfare ranking of competing policies when they make their decisions. They might also need to report the welfare gain. 

When the unconfoundedness assumption does not hold, identification of the conditional average treatment effect (CATE) and hence identification of the welfare gain becomes a delicate matter. Without further assumptions on selection, one cannot uniquely identify the welfare gain. I take a partial identification approach whereby one obtains bounds on the parameter of interest with a minimal amount of assumptions on the unobservables and, later on, tighten these bounds by imposing additional assumptions with and without instrumental variables. The bounds, or sharp identified region, of the welfare gain can be characterized using tools from random set theory.\footnote{ The terms \emph{identified region}, \emph{identified set}, and \emph{bounds} are used interchangeably throughout the paper. Often the word \emph{sharp} is omitted, and unless explicitly described as non-sharp, identified region/identified set/bounds refer to sharp identified region/sharp identified set/sharp bounds.} The framework I use allows me to consider various assumptions that involve instrumental variables and shape restrictions on the unobservables. 

I show that the lower and upper bounds of the welfare gain can, in general, be written as functions of the conditional mean treatment responses and a propensity score. Hence, estimation and inference of these bounds can be thought of as a semiparametric estimation problem in which the conditional mean treatment responses and the propensity score are infinite-dimensional nuisance parameters. Bounds that do not rely on instruments admit regular and asymptotically normal estimators. I construct orthogonalized, or locally robust, moment condition by adding an adjustment term that accounts for the first-step estimation to the original moment condition, following \cite{chernozhukov2020locally} (CEINR, henceforth). This method leads to estimators that are first-order insensitive to estimation errors of the nuisance parameters.  I calculate the adjustment term using an approach proposed by \cite{ichimura2017influence}. The locally robust estimation is possible even with instrumental variables under an additional monotonicity assumption of instruments. The estimation strategy has at least two advantages. First, it allows for flexible estimation of nuisance parameters, including the possibility of using high-dimensional machine learning methods. Second, the calculation of confidence intervals for the bounds is straightforward because the asymptotic variance doesn't rely on the estimation of nuisance parameters.

I illustrate the analysis using experimental data from the National Job Training Partnership Act (JTPA) Study. This dataset has been analyzed extensively in economics to understand the effect of subsidized training on outcomes such as earnings. I consider two hypothetical examples. First, I compare two different treatment assignment policies that are functions of individuals' years of education. Second, I compare \cite{kitagawa2018should}'s estimated optimal policy with an alternative policy when the conditioning variables are individuals' years of education and pre-program annual earnings. The results from a Monte Carlo simulation suggest that the method works well in a finite sample. 

\subsection{Related Literature} \label{subsection:literature}
This paper is related to the literature on treatment assignment, sometimes also referred to as treatment choice, which has been growing in econometrics since the seminal work by  \cite{manski2004statistical}. Earlier work in this literature include \cite{dehejia2005program}, \cite{hirano2009asymptotics}, \cite{stoye2009minimax, stoye2012minimax}, \cite{chamberlain2011bayesian}, \cite{bhattacharya2012inferring}, \cite{tetenov2012statistical}, \cite{kasy2014using}, and \cite{armstrong2015inference}. 

In a recent work, \cite{kitagawa2018should} propose what they call an empirical welfare maximization method. This method selects a treatment rule that maximizes the sample analog of the average social welfare over a class of candidate treatment rules. Their method has been further studied and extended in different directions. \cite{kitagawa2019equality} study an alternative welfare criterion that concerns equality. \cite{mbakop2016model} propose what they  call a penalized welfare maximization, an alternative method to estimate optimal treatment rules. While \cite{andrews2019inference} consider inference for the estimated optimal rule, \cite{rai2018statistical} considers inference for the optimal rule itself. These papers and most of the earlier papers only apply to a setting in which the assumption of unconfoundedness holds.

In a dynamic setting, treatment assignment is studied by \cite{kock2017optimal}, \cite{kock2018functional}, \cite{adusumilli2019dynamically}, \cite{sakaguchi2019estimating}, and \cite{han2019optimal}, among others.

This paper contributes to the less explored case of using observational data to infer policy choice where the unconfoundedness assumption does not hold. Earlier work in the treatment choice literature with partial identification include \cite{stoye2007minimax} and \cite{stoye2009partial}. This paper is closely related to \cite{kasy2016partial}, but their main object of interest is the welfare ranking of policies rather than the magnitude of welfare gain that results from switching from one policy to another policy. It is also closely related to \cite{athey2020policy} as they are concerned with choosing treatment assignment policies using observational data. However, their approach is about estimating the optimal treatment rule by point identifying the causal effect using various assumptions. In a related work in statistics, \cite{cui2020semiparametric} propose a method to estimate optimal treatment rules using instrumental variables. More recently, \cite{assunccao2019optimal} work with a partially identified welfare criterion that also takes spillover effects into account to analyze deforestation regulations in Brazil.

The rest of the paper is structured as follows. In Section \ref{section:setup}, I set up the problem. Section \ref{section:identification} presents the identification results of the welfare gain. Section \ref{section:estimation_and_inference} discusses the estimation and inference of the bounds.  In Section \ref{section:empirical_application}, I illustrate the analysis using experimental data from the National JTPA study. Section \ref{section:simulation} summarizes the results from a Monte Carlo simulation. Finally, Section \ref{section:conclusion} concludes.  All proofs, some useful definitions and theorems from random set theory, additional tables and figures from the empirical application, and more details on the simulation study are collected in the Appendix. 

\begin{notation} Throughout the paper, for $d\in\mathbb N$, let $\mathbb R^d$ denote the Euclidean space and $\lVert\cdot\rVert$ denote the Euclidean norm. Let $\langle \cdot, \cdot \rangle$ denote the inner product in $\mathbb R^d$ and $E[\cdot]$ denote the expectation operator. The notation $\overset{p}{\longrightarrow}$  and $\overset{d}{\longrightarrow}$ denote convergence in probability and convergence in distribution, respectively. For a sequence of numbers $x_n$ and $y_n$, $x_n= o(y_n)$ and $x_n=O(y_n)$ mean, respectively, that $x_n/y_n \to 0$ and $x_n \leq C y_n$ for some constant $C$ as $n \to \infty$. For a sequence of random variables $X_n$ and $Y_n$, the notation $X_n =o_p(Y_n)$ and $X_n =O_p(Y_n)$ mean, respectively, that $X_n/Y_n \overset{p}{\longrightarrow} 0$ and $X_n/Y_n$ is bounded in probability. $\mathcal N(\mu, \Omega)$ denotes a normal distribution with mean $\mu$ and variance $\Omega$. $\Phi(\cdot)$ denotes the cumulative distribution function of the standard normal distribution.
\end{notation}

\section{Setup} \label{section:setup}
Let $(\Omega,\mathfrak{A})$ be a measurable space. Let $Y:\Omega \to \mathbb R$ denote an outcome variable, $D:\Omega \to \{0,1\}$ denote a binary treatment, and $X: \Omega \to \mathcal X \subset \mathbb R^{d_x}$ denote pretreatment covariates. For $d\in\{0,1\}$, let $Y_d:\Omega \to \mathbb R$ denote a potential outcome that would have been observed if the treatment status were $D=d$. For each individual, the researcher only observes either $Y_1$ or $Y_0$ depending on what treatment the individual received. Hence, the relationship between observed and potential outcomes is given by
\begin{equation}
Y=Y_1 \cdot D + Y_0 \cdot(1-D).
\end{equation}
Policy I consider is a treatment assignment rule based on observed characteristics of individuals. In other words, the policymaker assigns an individual with covariate $X$ to a binary treatment according to a treatment rule $\delta: \mathcal{X} \rightarrow \{0,1\}$.\footnote{I consider deterministic treatment rules in my framework. See Appendix \ref{appendix:moregeneralcase} for discussions on randomized treatment rules.}  The welfare criterion considered is population mean welfare. If the policymaker chooses policy $\delta$, the welfare is given by
\begin{align} \label{equation:welfare}
\begin{split}
u(\delta) &\equiv E\big[Y_1\cdot\delta(X) + Y_0 \cdot(1-\delta(X))\big] \\
& = E\big[E[Y_1|X]\cdot\delta(X) + E[Y_0|X]\cdot(1-\delta(X))\big].
\end{split}
\end{align}
The object of my interest is \emph{welfare gain} that results from switching from policy $\delta^*$ to another policy $\delta$ which is
\begin{equation}
u(\delta) - u(\delta^*) = E\big[\Delta(X)\cdot(\delta(X)- \delta^*(X))\big], \indent \Delta(X) \equiv E[Y_1-Y_0|X].
\end{equation}

\begin{remark}
I assume that individuals comply with the assignment. This can serve as a natural baseline for choosing between policies. 
\end{remark}

The observable variables in my model are $(Y,D,X)$ and I assume that the researcher knows the joint distribution of  $(Y,D,X)$ when I study identification. Later, in Section \ref{section:estimation_and_inference}, I assume availability of data -- size $n$ random sample from $(Y,D,X)$ -- to conduct inference on objects that depend on this joint distribution. The unobservables in my model are potential outcomes $(Y_1,Y_0)$. The conditional average treatment effect $\Delta(X)=E[Y_1-Y_0|X]$ and hence my object of interest welfare gain cannot be point identified in the absence of strong assumptions. One instance in which it can be point identified is when potential outcomes $(Y_1,Y_0)$ are independent of treatment $D$ conditional on $X$, i.e., 
\begin{align}
(Y_1,Y_0)\perp D|X.
\end{align}
This assumption is called \emph{unconfoundedness} and is a widely-used identifying assumption in causal inference. See \cite{imbens2015causal} Chapter 12 and 21 for more discussions on this assumption. Under unconfoundedness, the conditional average treatment effect can be identified as 
\begin{equation} \label{equation:cate_point_id}
E[Y_1-Y_0|X] = E[Y|D=1, X]- E[Y|D=0, X].
\end{equation}
Note that the right-hand side of \eqref{equation:cate_point_id} is identified since the researcher knows the joint distribution of $(Y,D,X)$. If data are obtained from a randomized experiment, the assumption holds since the treatment is randomly assigned. However, if data are obtained from an observational study, the assumption is not testable and often controversial. In the next section, I relax the assumption of unconfoundedness and explore what can be learned about my parameter of interest when different assumptions are imposed on the unobservables and when there are additional instrumental variables $Z\in \mathcal Z \subset \mathbb R^{d_z}$ to help identify the conditional average treatment effect.

The welfare gain is related to \cite{manski2004statistical}'s \emph{regret} which has been used by \cite{kitagawa2018should}, \cite{athey2020policy}, and many others in the literature to evaluate the performance of the estimated treatment rules. When $\mathcal D$ is the class of treatment rules to be considered, the regret from choosing treatment rule $\delta$ is $u(\delta^*)-u(\delta)$ where 
\begin{equation} \label{equation:oracle}
\delta^* =  \arg\max_{d \in \mathcal D} E\big[E[Y_1|X]\cdot d + E[Y_0|X]\cdot(1-d)\big]. 
\end{equation}
It is an expected loss in welfare that results from not reaching the maximum feasible welfare as $\delta^*$ is the policy that maximizes population welfare. In \cite{kitagawa2018should} and others, under the assumption of unconfoundedness, the welfare criterion $u(\delta)$ in \eqref{equation:welfare} is point-identified. Therefore, the optimal "oracle" treatment rule in \eqref{equation:oracle} is well defined when the researcher knows the joint distribution of $(Y,D,X)$. However, when the welfare criterion in \eqref{equation:welfare} is set-identified, one needs to specify their notion of optimality. For instance, the optimal rule could be a rule that maximizes the guaranteed or minimum welfare.

\section{Identification} \label{section:identification}
\subsection{Sharp identified region}
Partial identification approach has been proven to be a useful alternative or complement to point identification analysis with strong assumptions. See \cite{manski2003partial}, \cite{tamer2010partial}, and \cite{molinari2019econometrics} for an overview. The theory of random sets, which I use to conduct my identification analysis, is one of the tools that have been used fruitfully to address identification and inference in partially identified models. Examples include \cite{beresteanu2008asymptotic}, \cite{beresteanu2011sharp, beresteanu2012partial}, \cite{galichon2011set}, \cite{epstein2016robust}, \cite{chesher2017generalized}, and \cite{kaido2019robust}. See \cite{molchanov2018random} for a textbook treatment of its use in econometrics.  

My goal in this section is to characterize the  \emph{sharp identified region} of the welfare gain when different assumptions are imposed on the unobservables. The sharp identified region of the welfare gain is the tightest possible set that collects the values of welfare gain that results from all possible $(Y_1,Y_0)$ that are consistent with the maintained assumptions. Toward this end, I define a \emph{random set} and its \emph{selections} whose formal definitions can be found in Appendix \ref{appendix:randomsettheory}. The random set is useful for incorporating weak assumptions in a unified framework rather than deriving bounds on a case-by-case basis. Let $(\mathcal{Y}_1 \times \mathcal{Y}_0): \Omega \to \mathcal F$ be a random set where $\mathcal F$ is the family of closed subsets of $\mathbb R^2$. Assumptions on potential outcomes can be imposed through this random set. Then, the collection of all random vectors $(Y_1, Y_0)$ that are consistent with those assumptions equals the family of all selections of $(\mathcal{Y}_1 \times \mathcal{Y}_0)$ denoted by $\mathcal S(\mathcal{Y}_1 \times \mathcal{Y}_0)$. Specific examples of a random set with more discussions on selections, namely, in the context of worst-case bounds of \cite{manski1990nonparametric} and monotone treatment response analysis of \cite{manski1997monotone}, are given in Section \ref{subsection:id_no_iv}. Using the random set notations I just introduced, the sharp identified region of the welfare gain is given by
\begin{equation} \label{equation:id}
B_I(\delta, \delta^*) \equiv
\{\beta \in \mathbb{R}: \beta = E\big[E[Y_1-Y_0|X]\cdot(\delta(X)- \delta^*(X))\big], (Y_1, Y_0) \in \mathcal S(\mathcal{Y}_1 \times \mathcal{Y}_0)\}.
\end{equation}

\subsection{Lower and upper bound}
One way to achieve characterization of the sharp identified region is through a \emph{selection expectation} and its \emph{support function}. Their definitions can be found in Appendix \ref{appendix:randomsettheory}. Let the support function of a convex set $K \subset \mathbb R^d$ be denoted by  
\begin{align}
s(v,K) = \sup_{x \in K}\langle v,x \rangle,  \indent v\in \mathbb R^d.
\end{align}
The support function appears in \cite{beresteanu2008asymptotic}, \cite{beresteanu2011sharp}, \cite{bontemps2012set}, \cite{kaido2014asymptotically}, \cite{kaido2016dual},  and \cite{kaido2017asymptotically}, among others.

I first state a lemma that will be useful to prove my main result. It shows how expectation of a functional of potential outcomes can be bounded from below and above by expected support function of the random set $(\mathcal{Y}_1 \times \mathcal{Y}_0)$. The proof of the following lemma and all other proofs in this paper are collected in the Appendix. 
\begin{lemma} \label{lemma:expectedsupport}
	Let $(\mathcal{Y}_1 \times \mathcal{Y}_0): \Omega \to \mathcal F$ be an integrable random set that is almost surely convex and let $(Y_1,Y_0)\in \mathcal S(\mathcal{Y}_1 \times \mathcal{Y}_0)$. For any $v\in \mathbb R^2$, we have
	\begin{equation} \label{expectedsupport}
		-E[s(-v,\mathcal{Y}_1 \times \mathcal{Y}_0)|X] \leq v'E[(Y_1, Y_0)'|X] \leq E[s(v,\mathcal{Y}_1 \times \mathcal{Y}_0)|X] \indent a.s.
	\end{equation}
\end{lemma}

I introduce a notation that appears in the following theorem and throughout the paper. Let $\theta_{10}(X) \equiv \mathbbm{1}\{\delta(X)=1, \delta^*(X)=0\}$ be an indicator function for the sub population that are newly treated under the new policy. Similary, let $\theta_{01}(X)\equiv\mathbbm{1}\{\delta(X)=0, \delta^*(X)=1\}$ be an indicator function for the sub population that are no longer being treated because of the new policy. 

\begin{theorem}[General case]\label{theorem:main}
	Suppose $(\mathcal{Y}_1 \times \mathcal{Y}_0): \Omega \to \mathcal F$ is an integrable random set that is almost surely convex. Let $\delta: \mathcal{X} \rightarrow \{0,1\}$ and $\delta^*: \mathcal{X} \rightarrow \{0,1\}$ be treatment rules.
	Also, let $v^*=(1,-1)'$. Then, $B_I(\delta, \delta^*)$ in \eqref{equation:id} is an interval $[\beta_l, \beta_u]$ where
\begin{equation} \label{equation:low_general}
	\beta_l = E[\underline{\Delta}(X)\cdot \theta_{10}(X) - \bar{\Delta}(X) \cdot \theta_{01}(X)],
\end{equation}
and
\begin{equation} \label{equation:up_general}
	\beta_u = E[\bar{\Delta}(X)\cdot\theta_{10}(X) -\underline{\Delta}(X)\cdot \theta_{01}(X)],
\end{equation}
	where $\underline{\Delta}(X)\equiv - E[s(-v^*,\mathcal{Y}_1 \times \mathcal{Y}_0)|X]$ and $\bar{\Delta}(X) \equiv E[s(v^*,\mathcal{Y}_1 \times \mathcal{Y}_0)|X]$.
\end{theorem}

The lower (upper) bound on the welfare gain is achieved when the newly treated people are the ones who benefit the least (most) from the treatment and the people who are no longer being treated are the ones who benefit the most (least) from the treatment. Therefore, the lower and upper bounds of the welfare gain involve both $\underline{\Delta}(X) = - E[s(-v^*,\mathcal{Y}_1 \times \mathcal{Y}_0)|X]$ and $\bar{\Delta}(X) = E[s(v^*,\mathcal{Y}_1 \times \mathcal{Y}_0)|X]$, expected support functions of the random set at directions $-v^*=(-1,1)'$ and $v^*=(1,-1)'$. Oftentimes, these can be estimated by its sample analog estimators. I give closed form expressions of the expected support functions in Section \ref{subsection:id_no_iv} and \ref{subsection:id_iv} -- they depend on objects such as $E[Y|D=1, X=x]$, $E[Y|D=0, X=x]$, and $P(D=1|X=x)$. To ease notation, let $\eta(d,x)\equiv E[Y|D=d, X=x]$ for $d\in\{0,1\}$ be the conditional mean treatment responses and $p(x) \equiv P(D=1|X=x)$ be the propensity score.

While I characterize the identified region of the welfare gain directly given assumptions on the selections $(Y_1, Y_0)$, \cite{kasy2016partial}'s analysis is based on the identified set for CATE and their main results apply to any approach that leads to partial identification of treatment effects. The characterization I give above is related to their characterization when no restrictions across covariate values are imposed on treatment effects (e.g., no restrictions such as $\Delta(x)$ is monotone in $x$) and $\underline{\Delta}(x)$  and $\bar{\Delta}(x)$ are respectively lower and upper bound on the CATE $\Delta(x)$. As examples of such bounds, \cite{kasy2016partial} considers bounds that arise under instrument exogeneity as in  \cite{manski2003partial} and under marginal stationarity of unobserved heterogeneity in panel data models as in \cite{chernozhukov2013average}. I consider bounds when there are instrumental variables that satisfy mean independence or mean monotonicity conditions as in \cite{manski2003partial} in Section \ref{subsection:id_iv}.	

In the following subsection, Section \ref{subsection:id_no_iv}, I illustrate the form of the random set and show how Theorem \ref{theorem:main} can be used to derive closed form bounds under different sets of assumptions.

\subsection{Identification without Instruments} \label{subsection:id_no_iv}

\begin{figure}[htbp]
	\begin{center}	
		\begin{tikzpicture}[scale=0.6]
		[domain=-0.5:6]	
		\node at (3,7) {$D=1$};
		\draw[->] (-0.5,0) -- (6,0) node[right] {$\mathcal Y_1$};
		\draw[->] (0,-0.5) -- (0,6) node[above] {$\mathcal Y_0$};
		\draw (3,0) node[below]{$Y$};
		\draw (0,2) node[left]{$\underline y$};
		\draw (0,5) node[left]{$\bar y$};
		\draw[dashed, gray] (0,2) -- (3,2);
		\draw[dashed, gray] (0,5) -- (3,5);
		\draw[dashed, gray] (3,0) -- (3,2);
		\draw[ultra thick, dark-blue] (3,2) -- (3,5);
		\end{tikzpicture}
		\hspace{0.2in}
		\begin{tikzpicture}[scale=0.6]
		[domain=-0.5:6]
		\node at (3,7) {$D=0$};
		\draw[->] (-0.5,0) -- (6,0) node[right] {$\mathcal Y_1$};
		\draw[->] (0,-0.5) -- (0,6) node[above] {$\mathcal Y_0$};
		\draw (0,3) node[left]{$Y$};
		\draw (2,0) node[below]{$\underline y$};
		\draw (5,0) node[below]{$\bar y$};
		\draw[dashed, gray] (0,3) -- (2,3);
		\draw[dashed, gray] (2,0) -- (2,3);
		\draw[dashed, gray] (5,0) -- (5,3);
		\draw[ultra thick, dark-blue] (2,3) -- (5,3);
		\end{tikzpicture}
	\end{center}
	\caption{Random set $(\mathcal{Y}_1 \times \mathcal{Y}_0)$ under worst-case}
	\label{figure:randomset_worstcase}
\end{figure}
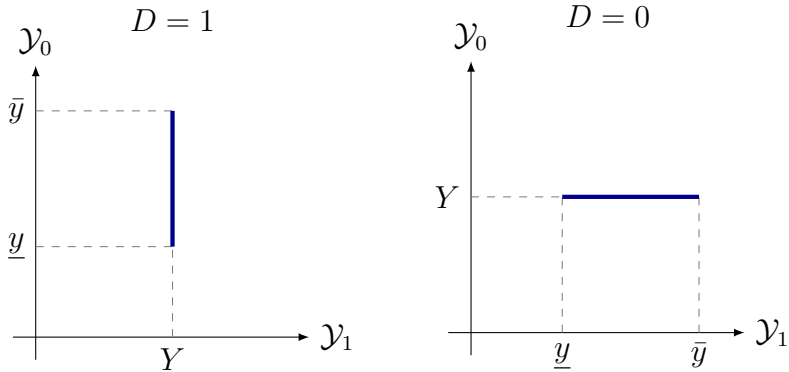

\noindent \cite{manski1990nonparametric} derived worst-case bounds on $Y_1$ and $Y_0$ when the outcome variable is bounded, i.e., $Y\in [\underline y, \bar y]\subset \mathbb R$ where $ -\infty<\underline y \leq \bar y < \infty$. It is called worst-case bounds because no additional assumptions are imposed on their distributions. Then, as shown in Figure \ref{figure:randomset_worstcase}, the random set $(\mathcal{Y}_1 \times \mathcal{Y}_0)$ is such that 
\begin{equation} \label{equation:randomset_worstcase}
\mathcal{Y}_1 \times \mathcal{Y}_0 = 
\begin{cases}
\{Y\}\times [\underline y, \bar y] $ if $ D=1, \\
[\underline y, \bar y] \times \{Y\} $ if $ D=0.
\end{cases}
\end{equation}

The random set in \eqref{equation:randomset_worstcase} switches its value between two sets depending on the value of $D$. If $D=1$, $\mathcal{Y}_1$ is given by a singleton $\{Y\}$ whereas $\mathcal{Y}_0$ is given by the entire support $[\underline y, \bar y]$. Similarly, if $D=0$, $\mathcal{Y}_0$ is given by a singleton $\{Y\}$ whereas $\mathcal{Y}_1$ is given by the entire support $[\underline y, \bar y]$. I plot  $\mathcal{Y}_d$ and its selection $Y_d$ for $d\in\{0,1\}$ as a function of $\omega \in \Omega$ in Figure \ref{figure:selection_worstcase}. If $D=d$, the random set $\mathcal{Y}_d$ is a singleton $\{Y\}$ and the family of selections consists of single random variable $\{Y\}$ as well. On the other hand, if $D=1-d$, the random set $\mathcal{Y}_d$ is an interval $[\underline y, \bar y]$ and the family of all selections consists of all $\mathfrak A$-measurable random variables that has support on $[\underline y, \bar y]$. 
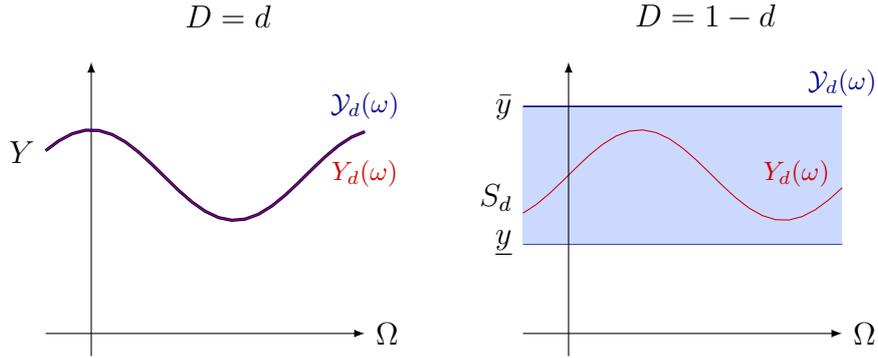
\begin{figure}[htbp]
	\begin{center}	
		\begin{tikzpicture}[scale=0.6]
		[domain=-1:6]
		\node at (3,7) {$D=d$};
		\draw[->] (-1,0) -- (6,0) node[right] {$\Omega$};
		\draw[->] (0,-0.5) -- (0,6) node[above] {};
		\draw (-1,4) node[left]{$Y$};
		\draw[very thick, dark-blue, domain =-1:6] plot(\x,{3.5+cos(\x r )});
		\draw[bu-red, domain =-1:6] plot(\x,{3.5+cos(\x r )});
		\draw (6,4.5) node[dark-blue, above]{\footnotesize{$\mathcal{Y}_d(\omega)$}};
		\draw (6,3) node[bu-red, above]{\footnotesize{$Y_d(\omega)$}};
		\end{tikzpicture}
		\hspace{0.2in}
		\begin{tikzpicture}[scale=0.6]
		[domain=-1:6]	
		\node at (3,7) {$D=1-d$};
		\draw (-1,2) node[left]{$\underline y$};
		\draw (-1,5) node[left]{$\bar y$};
		\draw (-1,3) node[left]{$S_d$};
		\draw[very thick, dark-blue] (-1,5) -- (6,5);
		\draw[very thick, dark-blue] (-1,2) -- (6,2);
		\fill[fill=light-blue] (-1,5) -- (6,5) -- (6,2) -- (-1,2);
		\draw[bu-red, domain =-1:6] plot(\x,{3.5+sin(\x r )});
		\draw (6,5) node[dark-blue, above]{\footnotesize{$\mathcal{Y}_d(\omega)$}};
		\draw (5,3) node[bu-red, above]{\footnotesize{$Y_d(\omega)$}};
		\draw[->] (-1,0) -- (6,0) node[right] {$\Omega$};
		\draw[->] (0,-0.5) -- (0,6) node[above] {};
		\end{tikzpicture}	
	\end{center}
	\caption{Random set $\mathcal{Y}_d$ and its selection $Y_d$ for $d\in\{0,1\}$ as a function of $\omega\in\Omega$ under worst-case}
	\label{figure:selection_worstcase}
\end{figure}
Note that each selection $(Y_1,Y_0)$ of $(\mathcal{Y}_1 \times \mathcal{Y}_0)$ can be represented in the following way. Take random variables $S_1:\Omega \to \mathbb R$ and $S_0:\Omega \to \mathbb R$ whose distributions conditional on $Y$ and $D$ are not specified and can be any probability distributions on $[\underline y, \bar y]$. Then $(Y_1, Y_0)$ that satisfies the following is a selection of $\mathcal{Y}_1 \times \mathcal{Y}_0$:
\begin{align} \label{equation:selection}
\begin{split}
Y_1 = Y\cdot D + S_1\cdot (1-D), \\
Y_0 = Y\cdot(1-D) + S_0\cdot D.
\end{split}
\end{align}
This representation makes it even clearer how I am not imposing any structure on the counterfactuals that I do not observe. $S_1$ and $S_0$ correspond to the \emph{selection mechanisms} that appear in \cite{ponomareva2011misspecification} and \cite{tamer2010partial}.

Now, for the random set in \eqref{equation:randomset_worstcase}, I can calculate its expected support function at directions $v^*=(1,-1)$ and $-v^*=(-1,1)$ to obtain the bounds of the welfare gain in closed form. As shown in Figure \ref{figure:supportfunction}, the support function of random set $(\mathcal{Y}_1 \times \mathcal{Y}_0)$ in \eqref{equation:randomset_worstcase} at direction $v^*=(1,-1)$ is the (signed) distance (rescaled by the norm of $v^*$) between the origin and the hyperplane tangent to the random set in direction $v^*=(1,-1)$. Then, the bounds are given in the following Corollary to Theorem \ref{theorem:main}.

\begin{figure}[htbp] 
	\begin{center}
		\begin{tikzpicture}[scale=0.8]
		[domain=-0.5:4.5]	
		\node at (2,5) {$D=1$};
		\draw[->] (-0.5,0) -- (4.5,0) node[right] {$\mathcal Y_1$};
		\draw[->] (0,-0.5) -- (0,4.5) node[above] {$\mathcal Y_0$};
		\draw[->, thick, bu-red] (0,0) -- (1,-1);
		\draw (1.5,-1.5) node[below, bu-red]{\footnotesize $v^*=(1,-1)$};
		\draw[very thin, gray] (-1,-1) -- (1,1);
		\draw[very thin, gray] (0,-1) -- (3,2);
		\draw[very thin, white] (0,-2) -- (5,3);
		\draw[->, dashed, thick, dark-blue] (1,1) -- (1.5,0.5);
		\draw (2,0) node[below]{$Y$};
		\draw (0,1) node[left]{$\underline y$};
		\draw (0,4) node[left]{$\bar y$};
		\draw[dashed, gray] (0,1) -- (2,1);
		\draw[dashed, gray] (0,4) -- (2,4);
		\draw[dashed, gray] (2,0) -- (2,1);
		\draw[ultra thick, dark-blue] (2,1) -- (2,4);
		\end{tikzpicture}
		\hspace{0.2in}
		\begin{tikzpicture}[scale=0.8]
		[domain=-0.5:4.5]
		\node at (2,5) {$D=0$};
		\draw[->] (-0.5,0) -- (4.5,0) node[right] {$\mathcal Y_1$};
		\draw[->] (0,-0.5) -- (0,4.5) node[above] {$\mathcal Y_0$};
		\draw[->, thick, bu-red] (0,0) -- (1,-1);
		\draw (1.5,-1.5) node[below, bu-red]{\footnotesize $v^*=(1,-1)$};
		\draw[very thin, gray] (-1,-1) -- (2,2);
		\draw[very thin, gray] (0,-2) -- (5,3);
		\draw[->, dashed,thick, dark-blue] (1,1) -- (2,0);
		\draw (0,2) node[left]{$Y$};
		\draw (1,0) node[below]{$\underline y$};
		\draw (4,0) node[below]{$\bar y$};
		\draw[dashed, gray] (0,2) -- (1,2);
		\draw[dashed, gray] (1,0) -- (1,2);
		\draw[dashed, gray] (4,0) -- (4,2);
		\draw[ultra thick, dark-blue] (1,2) -- (4,2);
		\end{tikzpicture}
	\end{center}
\caption{Support function of $(\mathcal{Y}_1 \times \mathcal{Y}_0)$ at direction $v^*=(1,-1)$ under worst-case}
\label{figure:supportfunction}
\end{figure}
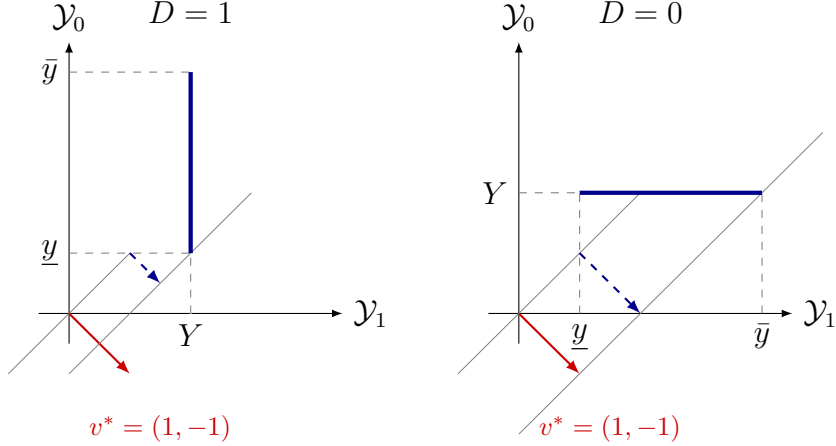

\begin{corollary}[Worst-case] \label{corollary:worstcase}
	Let $(\mathcal{Y}_1 \times \mathcal{Y}_0)$ be a random set in \eqref{equation:randomset_worstcase}.
	Let $\delta: \mathcal{X} \rightarrow \{0,1\}$ and $\delta^*: \mathcal{X} \rightarrow \{0,1\}$ be treatment rules.
	Then, $B_I(\delta, \delta^*)$ in \eqref{equation:id} is an interval $[\beta_l, \beta_u]$ where
	\begin{align} \label{equation:low_worstcase}
	\begin{split}	
	\beta_l & = E\big[\big((\eta(1,X)-\bar{y})\cdot p(X) + (\underline{y}-\eta(0,X))\cdot(1-p(X))\big)\cdot\theta_{10}(X)\\
	&\qquad - \big((\eta(1,X)-\underline{y})\cdot p(X) + (\bar{y}-\eta(0,X))\cdot(1-p(X))\big) \cdot\theta_{01}(X)\big],
	\end{split}
	\end{align}
	and
	\begin{align} \label{equation:up_worstcase}
	\begin{split}
	\beta_u & = E\big[\big((\eta(1,X)-\underline{y})\cdot p(X) + (\bar{y}-\eta(0,X))\cdot(1-p(X))\big)\cdot \theta_{10}(X)\\
	& \qquad -\big((\eta(1,X)-\bar{y})\cdot p(X) + (\underline{y}-\eta(0,X))\cdot(1-p(X))\big)\cdot \theta_{01}(X)\big].
	\end{split}
	\end{align}
	
\end{corollary}

\noindent Worst-case analysis is a great starting point as no additional assumptions are imposed on the unobservables. However, the bounds could be too wide to be informative in some cases. In fact, the worst-case bound cover $0$ all the time as $\underline \beta \leq0$ and $\beta_u \geq0$. One could impose additional assumptions on the relationship between the unobservables and obtain tighter bounds. Towards that end, I analyze the monotone treatment response (MTR) assumption of \cite{manski1997monotone}. 
\begin{assumption}[MTR Assumption] \label{assumption:mtr}
\begin{equation} 
Y_1 \geq Y_0 \text{ a.s.}
\end{equation}
\end{assumption}

Assumption \ref{assumption:mtr} states that everyone benefits from the treatment. Suppose Assumption \ref{assumption:mtr} holds. Then, the random set is such that
\begin{equation} \label{equation:randomset_mtr}
\mathcal{Y}_1 \times \mathcal{Y}_0 = 
\begin{cases}
\{Y\}\times [\underline y, Y] $ if $ D=1, \\
[Y, \bar y] \times \{Y\} $ if $ D=0.
\end{cases}
\end{equation}

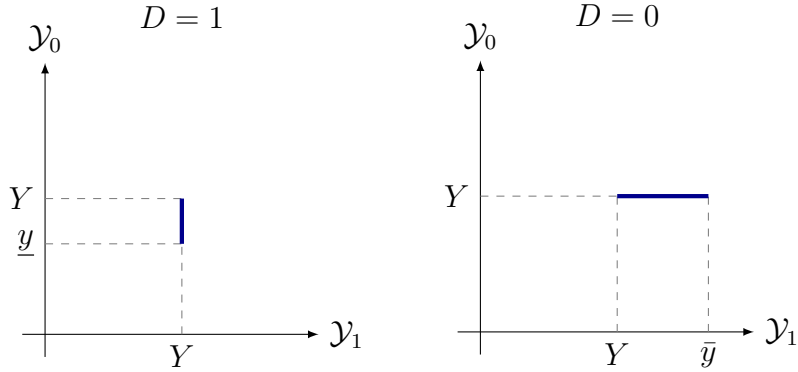
\begin{figure}[h] 
	\begin{center}	
		\begin{tikzpicture}[scale=0.6]
		[domain=-0.5:6]	
		\node at (3,7) {$D=1$};
		\draw[->] (-0.5,0) -- (6,0) node[right] {$\mathcal Y_1$};
		\draw[->] (0,-0.5) -- (0,6) node[above] {$\mathcal Y_0$};
		\draw (3,0) node[below]{$Y$};
		\draw (0,2) node[left]{$\underline y$};
		\draw (0,3) node[left]{$Y$};
		\draw[dashed, gray] (0,2) -- (3,2);
		\draw[dashed, gray] (0,3) -- (3,3);
		\draw[dashed, gray] (3,0) -- (3,2);
		\draw[ultra thick, dark-blue] (3,2) -- (3,3);
		\end{tikzpicture}
		\hspace{0.2in}
		\begin{tikzpicture}[scale=0.6]
		[domain=-0.5:6]
		\node at (3,7) {$D=0$};
		\draw[->] (-0.5,0) -- (6,0) node[right] {$\mathcal Y_1$};
		\draw[->] (0,-0.5) -- (0,6) node[above] {$\mathcal Y_0$};
		\draw (0,3) node[left]{$Y$};
		\draw (3,0) node[below]{$Y$};
		\draw (5,0) node[below]{$\bar y$};
		\draw[dashed, gray] (0,3) -- (3,3);
		\draw[dashed, gray] (3,0) -- (3,3);
		\draw[dashed, gray] (5,0) -- (5,3);
		\draw[ultra thick, dark-blue] (3,3) -- (5,3);
		\end{tikzpicture}
	\end{center}
	\caption{Random set $(\mathcal{Y}_1 \times \mathcal{Y}_0)$ under MTR Assumption}
	\label{figure:randomset_mtr}
\end{figure}

As shown in Figure \ref{figure:randomset_mtr}, depending on the value of $D$, the random set in \eqref{equation:randomset_mtr} switches its value between two sets, that are smaller than those in \eqref{equation:randomset_worstcase}. The bounds of the welfare gain when the random set is given by \eqref{equation:randomset_mtr} are given in the following Corollary to Theorem \ref{theorem:main}. Notice that the lower bound on conditional average treatment effect $\underline{\Delta}(X) = - E[s(-v^*,\mathcal{Y}_1 \times \mathcal{Y}_0)|X]$ equals $0$ when the random set is given by \eqref{equation:randomset_mtr}. It is shown geometrically in Figure \ref{figure:supportfunction_mtr}. The expected support function of the random set in \eqref{equation:randomset_mtr} at direction $-v^*=(-1,1)'$ is always $0$ as the hyperplane tangent to the random set at direction $-v^*=(-1,1)'$ goes through the origin regardless of the value of $D$. 

\begin{figure}[h] 
	\begin{center}	
		\begin{tikzpicture}[scale=0.6]
		[domain=-0.5:6]	
		\node at (3,7) {$D=1$};
		\draw[->] (-0.5,0) -- (6,0) node[right] {$\mathcal Y_1$};
		\draw[->] (0,-0.5) -- (0,6) node[above] {$\mathcal Y_0$};
		\draw (3,0) node[below]{$Y$};
		\draw (0,2) node[left]{$\underline y$};
		\draw (0,3) node[left]{$Y$};
		\draw[dashed, gray] (0,2) -- (3,2);
		\draw[dashed, gray] (0,3) -- (3,3);
		\draw[dashed, gray] (3,0) -- (3,2);
		\draw[ultra thick, dark-blue] (3,2) -- (3,3);
		\draw[->, thick, bu-red] (0,0) -- (-1,1);
		\draw (-2.5,2) node[below, bu-red]{\footnotesize $-v^*=(-1,1)$};
		\draw[very thin, gray] (-1,-1) -- (4,4);
		\end{tikzpicture}
		\hspace{0.2in}
		\begin{tikzpicture}[scale=0.6]
		[domain=-0.5:6]
		\node at (3,7) {$D=0$};
		\draw[->] (-0.5,0) -- (6,0) node[right] {$\mathcal Y_1$};
		\draw[->] (0,-0.5) -- (0,6) node[above] {$\mathcal Y_0$};
		\draw (0,3) node[left]{$Y$};
		\draw (3,0) node[below]{$Y$};
		\draw (5,0) node[below]{$\bar y$};
		\draw[dashed, gray] (0,3) -- (3,3);
		\draw[dashed, gray] (3,0) -- (3,3);
		\draw[dashed, gray] (5,0) -- (5,3);
		\draw[ultra thick, dark-blue] (3,3) -- (5,3);
		\draw[->, thick, bu-red] (0,0) -- (-1,1);
		\draw (-2.5,2) node[below, bu-red]{\footnotesize $-v^*=(-1,1)$};
		\draw[very thin, gray] (-1,-1) -- (1,1);
		\draw[very thin, gray] (-1,-1) -- (4,4);
		\end{tikzpicture}
	\end{center}
	\caption{Support function of $(\mathcal{Y}_1 \times \mathcal{Y}_0)$ at direction $-v^*=(-1,1)'$ under MTR Assumption}
	\label{figure:supportfunction_mtr}
\end{figure}
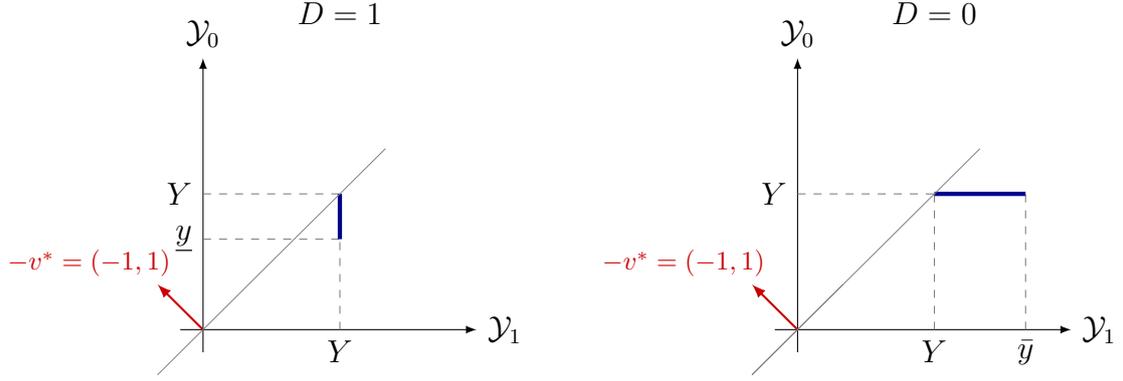

\begin{corollary}[MTR] \label{corollary:mtr}
	Suppose Assumption \ref{assumption:mtr} holds. Let $\delta: \mathcal{X} \rightarrow \{0,1\}$ and $\delta^*: \mathcal{X} \rightarrow \{0,1\}$ be treatment rules. Then, $B_I(\delta, \delta^*)$ in \eqref{equation:id} is an interval $[\beta_l, \beta_u]$ where
	\begin{align} \label{equation:low_mtr}
	\beta_l = E\big[ -\big((\eta(1,X)-\underline{y})\cdot p(X) + (\bar{y}-\eta(0,X))\cdot(1-p(X))\big)\cdot\theta_{01}(X)\big],
	\end{align}
	and
	\begin{align} \label{equation:up_mtr}
	\beta_u = E\big[\big((\eta(1,X)-\underline{y})\cdot p(X) + (\bar{y}-\eta(0,X))\cdot(1-p(X))\big)\cdot\theta_{10}(X)\big].
	\end{align}
\end{corollary}

\subsection{Identification with Instruments} \label{subsection:id_iv}
Availability of additional variables, called \emph{instrumental variables}, could help us tighten the bounds on CATE and hence the bounds on the welfare gain. In this subsection, I consider two types of assumptions: (1) mean independence (IV Assumption) and (2) mean monotonicity (MIV Assumption).

\subsubsection{Mean independence}
\begin{assumption}[IV Assumption]
 \label{assumption:iv}	
	There exists an instrumental variable $Z\in\mathcal Z \subset \mathbb R^{d_z}$ such that, for $d\in\{0,1\}$, the following mean-independence holds:
	\begin{equation} \label{equation:iv}
	E[Y_d|X,Z=z]=E[Y_d|X,Z=z'],
	\end{equation}
	for all $z,z' \in \mathcal Z.$
\end{assumption}

\noindent When data are obtained from a randomized experiment with imperfect compliance, the random assignment can be used as an instrumental variable to identify the effect of the treatment. 

Suppose Assumption \ref{assumption:iv} holds. Since I am imposing an additional restriction on $(Y_1, Y_0)$, the sharp identified region of the  welfare gain is given by 
\begin{align} \label{equation:id_iv}
\begin{split}
B_I(\delta, \delta^*) \equiv &
\{\beta \in \mathbb{R}: \beta = E\big[E[Y_1-Y_0|X]\cdot(\delta(X)- \delta^*(X))\big], (Y_1, Y_0) \in \mathcal S(\mathcal{Y}_1 \times \mathcal{Y}_0), \\
& \qquad \qquad (Y_1, Y_0) \text{ satisfies Assumption \ref{assumption:iv}}\}.
\end{split}
\end{align} 
The following lemma corresponds to the Manski's sharp bounds for CATE under mean-independence assumption. \cite{manski1990nonparametric} explains it for the more general case of when there are level-set restrictions on the outcome regression. 
\begin{lemma}[IV]\label{lemma:expectedsupport_iv}
	Let $(\mathcal{Y}_1 \times \mathcal{Y}_0): \Omega \to \mathcal F$ be an integrable random set that is almost surely convex and let $(Y_1,Y_0)\in \mathcal S(\mathcal{Y}_1 \times \mathcal{Y}_0)$. Let $v_1=(1,0)'$ and $v_0=(0,1)'$. Suppose Assumption \ref{assumption:iv} holds. Then, we have
	\begin{equation}  \label{equation:iv_delta}
	\begin{split}
	&\sup_{z \in\mathcal Z} \big\{ -E[s(-v_1,\mathcal{Y}_1 \times \mathcal{Y}_0)|X, Z=z]\big\}- \inf_{z\in\mathcal Z} \big\{ E[s(v_0,\mathcal{Y}_1 \times \mathcal{Y}_0)|X, Z=z]\big\} \\
	& \qquad\qquad\qquad\qquad\qquad\qquad \leq E[Y_1-Y_0|X] \leq \\
	&  \inf_{z\in\mathcal Z} \big\{  E[s(v_1,\mathcal{Y}_1 \times \mathcal{Y}_0)|X, Z=z] \big\}- \sup_{z\in\mathcal Z}\big\{ -E[s(-v_0,\mathcal{Y}_1 \times \mathcal{Y}_0)|X, Z=z]\big\}\indent a.s.
	\end{split}
	\end{equation}
\end{lemma}

Bounds for CATE with instrumental variables involve expected support functions at directions $v_1=(1,0)$ and $v_0=(0,1)$. The support function of the random set $(\mathcal{Y}_1 \times \mathcal{Y}_0)$ at direction $v_1=(1,0)$ under worst-case is depicted in Figure \ref{figure:supportfunction_10}.

\begin{figure}[htbp] 
	\begin{center}
		\begin{tikzpicture}[scale=0.8]
		[domain=-0.5:5.5]	
		\node at (3,6) {$D=1$};
		\draw[->] (-0.5,0) -- (5.5,0) node[right] {$\mathcal Y_1$};
		\draw[->] (0,-0.5) -- (0,5.5) node[above] {$\mathcal Y_0$};
		\draw[->, very thick, bu-red] (0,0) -- (1,0);
		\node[bu-red] at (1,-0.5) {\scriptsize $v_1=(1,0)$};
		\draw[very thin, gray] (3,0) -- (3,5);
		\draw[very thin, white] (0,-2) -- (5,3);
		\draw[->, thick, dashed,thick, dark-blue] (0,3) -- (3,3);
		\draw (3,0) node[below]{$Y$};
		\draw (0,2) node[left]{$\underline y$};
		\draw (0,5) node[left]{$\bar y$};
		\draw[dashed, gray] (0,2) -- (3,2);
		\draw[dashed, gray] (0,5) -- (3,5);
		\draw[dashed, gray] (3,0) -- (3,5);
		\draw[ultra thick, dark-blue] (3,2) -- (3,5);
		\end{tikzpicture}
		\hspace{0.2in}
		\begin{tikzpicture}[scale=0.8]
		[domain=-0.5:5.5]
		\node at (3,6) {$D=0$};
		\draw[->] (-0.5,0) -- (5.5,0) node[right] {$\mathcal Y_1$};
		\draw[->] (0,-0.5) -- (0,5.5) node[above] {$\mathcal Y_0$};
		\draw[->, very thick, bu-red] (0,0) -- (1,0);
		\node[bu-red] at (1,-0.5) {\scriptsize $v_1=(1,0)$};
		\draw[very thin] (0,0) -- (0,4);
		\draw[very thin, gray] (5,0) -- (5,5);
		\draw[very thin, white] (0,-2) -- (5,3);
		\draw[->, dashed, thick, dark-blue] (0,1) -- (5,1);
		\draw (0,3) node[left]{$Y$};
		\draw (2,0) node[below]{$\underline y$};
		\draw (5,0) node[below]{$\bar y$};
		\draw[dashed, gray] (0,3) -- (2,3);
		\draw[dashed, gray] (2,0) -- (2,3);
		\draw[dashed, gray] (5,0) -- (5,3);
		\draw[ultra thick, dark-blue] (2,3) -- (5,3);
		\end{tikzpicture}
	\end{center}
\caption{Support function of $(\mathcal{Y}_1 \times \mathcal{Y}_0)$ at direction $v_1=(1,0)$ under worst-case}
\label{figure:supportfunction_10}
\end{figure}
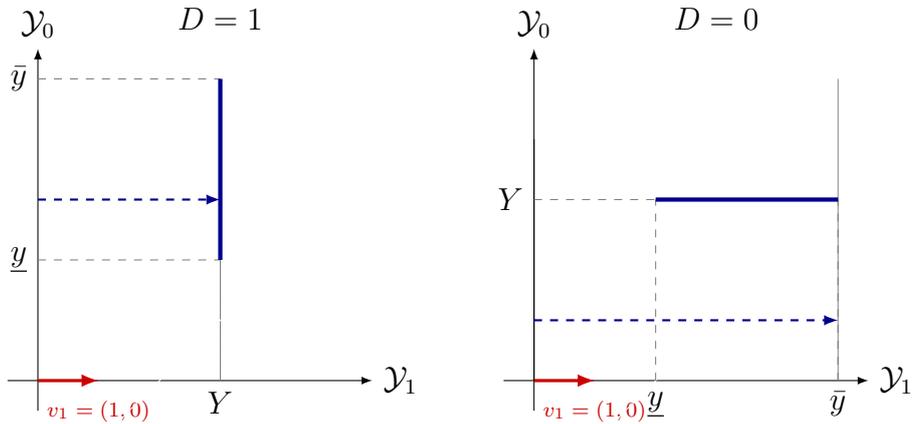

\begin{theorem}[IV]  \label{theorem:iv}
	Suppose $(\mathcal{Y}_1 \times \mathcal{Y}_0): \Omega \to \mathcal F$ is an integrable random set that is almost surely convex.  Let $\delta: \mathcal{X} \rightarrow \{0,1\}$ and $\delta^*: \mathcal{X} \rightarrow \{0,1\}$ be treatment rules. Also, let $v_1=(1,0)'$ and $v_0=(0,1)'$. Then, $B_I(\delta, \delta^*)$ in \eqref{equation:id_iv} is an interval $[\beta_l, \beta_u]$ where
	\begin{equation} \label{lower}
	\beta_l = E\big[\underline{\Delta}(X) \cdot \theta_{10}(X) - \bar{\Delta}(X) \cdot \theta_{01}(X)\big],
	\end{equation}
	and
	\begin{equation} \label{upper}
	\beta_u = E\big[\bar{\Delta}(X)\cdot \theta_{10}(X) -\underline{\Delta}(X)\cdot \theta_{01}(X)\big],
	\end{equation}
	where $\underline{\Delta}(X)\equiv \sup_{z \in\mathcal Z} \big\{ -E[s(-v_1,\mathcal{Y}_1 \times \mathcal{Y}_0)|X, Z=z]\big\}- \inf_{z\in\mathcal Z} \big\{ E[s(v_0,\mathcal{Y}_1 \times \mathcal{Y}_0)|X, Z=z]\big\}$ and $\bar{\Delta}(X) \equiv \inf_{z\in\mathcal Z} \big\{  E[s(v_1,\mathcal{Y}_1 \times \mathcal{Y}_0)|X, Z=z] \big\}- \sup_{z\in\mathcal Z}\big\{ -E[s(-v_0,\mathcal{Y}_1 \times \mathcal{Y}_0)|X, Z=z]\big\}$.
\end{theorem}

Identification of the welfare gain with instruments is similar to idenfication without instruments. The difference lies in the forms of lower and upper bounds on the CATE. Theorem \ref{theorem:iv} can be combined with different maintained assumptions on the potential outcomes to result in different bounds. Corollary \ref{corollary:iv_worstcase} shows the IV bounds under worst-case assumption and Corollary \ref{corollary:iv_mtr} shows the IV bounds under MTR assumption. To ease notation, let $\eta(d,x,z)\equiv E[Y|D=d, X=x, Z=z]$ for $d\in\{0,1\}$ denote the conditional mean treatment responses and $p(x,z) \equiv P(D=1|X=x, Z=z)$ denote the propensity score.

\begin{corollary}[IV-worst case]\label{corollary:iv_worstcase}
	Let $(\mathcal{Y}_1 \times \mathcal{Y}_0)$ be a random set in \eqref{equation:randomset_worstcase}. Let $\delta: \mathcal{X} \rightarrow \{0,1\}$ and $\delta^*: \mathcal{X} \rightarrow \{0,1\}$ be treatment rules. Then, $B_I(\delta, \delta^*)$ in \eqref{equation:id_iv} is an interval $[\beta_l, \beta_u]$ where
	\begin{align} \label{equation:low_manski-iv_worstcase}
	\begin{split}
	\beta_l & = E\big[\big(\sup_{z\in\mathcal Z} \big\{ \eta(1, X, z)\cdot p(X, z) + \underline{y}\cdot(1-p(X, z))\big\} \\
	& \qquad\qquad- \inf_{z\in\mathcal Z} \big\{\bar{y}\cdot p(X,z) + \eta(0, X,z)\cdot(1-p(X,z))\big\}\big) \cdot \theta_{10}(X) \\
	& \qquad- \big(\inf_{z\in\mathcal Z} \big\{\eta(1, X, z)\cdot p(X, z) + \bar{y}\cdot(1-p(X,z))\big\} \\
	& \qquad\qquad- \sup_{z\in\mathcal Z} \big\{\underline{y}\cdot p(X,z) + \eta(0, X,z)\cdot(1-p(X,z))\big\}\big) \cdot \theta_{01}(X)\big],
	\end{split}
	\end{align}
	and
	\begin{align} \label{equation:up_manski-iv_worstcase}
	\begin{split}
	\beta_u & = E\big[\big(\inf_{z\in\mathcal Z} \big\{\eta(1, X, z)\cdot p(X, z) + \bar{y}\cdot(1-p(X,z))\big\} \\
	& \qquad\qquad- \sup_{z\in\mathcal Z} \big\{\underline{y}\cdot p(X,z) + \eta(0, X,z)\cdot(1-p(X,z))\big\}\big)\cdot \theta_{10}(X) \\
	& \qquad- \big(\sup_{z\in\mathcal Z} \big\{ \eta(1, X, z)\cdot p(X, z) + \underline{y}\cdot(1-p(X, z))\big\} \\
	& \qquad \qquad - \inf_{z\in\mathcal Z} \big\{\bar{y}\cdot p(X,z) + \eta(0, X,z)\cdot(1-p(X,z))\big\}\big) \cdot \theta_{01}(X)\big].
	\end{split}
	\end{align}	
\end{corollary}

\begin{corollary}[IV-MTR]\label{corollary:iv_mtr}
	Suppose Assumption \ref{assumption:mtr} holds. Let $\delta: \mathcal{X} \rightarrow \{0,1\}$ and $\delta^*: \mathcal{X} \rightarrow \{0,1\}$ be treatment rules. Then, $B_I(\delta, \delta^*)$ in \eqref{equation:id_iv} is an interval $[\beta_l, \beta_u]$ where	
	\begin{align} \label{equation:low_manski-iv_mtr}
	\begin{split}
	\beta_l & = E\big[-\big(\inf_{z\in\mathcal Z} \big\{\eta(1, X, z)\cdot p(X, z) + \bar{y}\cdot(1-p(X,z))\big\} \\
	& \qquad\qquad- \sup_{z\in\mathcal Z} \big\{\underline{y}\cdot p(X,z) + \eta(0, X,z)\cdot(1-p(X,z))\big\}\big) \cdot \theta_{01}(X)\big],
	\end{split}
	\end{align}
	and
	\begin{align} \label{equation:up_manski-iv_mtr}
	\begin{split}
	\beta_u & = E\big[\big(\inf_{z\in\mathcal Z} \big\{\eta(1, X, z)\cdot p(X, z) + \bar{y}\cdot(1-p(X,z))\big\} \\
	& \qquad\qquad- \sup_{z\in\mathcal Z} \big\{\underline{y}\cdot p(X,z) + \eta(0, X,z)\cdot(1-p(X,z))\big\}\big)\cdot \theta_{10}(X)\big].
	\end{split}
	\end{align}	
\end{corollary}

Bounds obtained with instruments are functions of $\eta(1,x,z)$, $\eta(0,x,z)$ and $p(x,z)$ and involve taking intersections across values of $Z$. If $Z$ is continuous, this would amount to infinitely many intersections. However, bounds can be simplified in some empirically relevant cases such as the following. 

\begin{assumption}[Binary IV with monotonic first-step] \label{assumption:binary_iv}
Suppose $Z\in\{0,1\}$ is a binary instrumental variable that satisfies Assumption \ref{assumption:iv}. Suppose further that for all $x \in \mathcal X$,
\begin{align}\label{equation:simulation_iv1}
p(x,1)=P(D=1|X=x, Z=1) \geq P(D=1|X=x, Z=0) = p(x,0).
\end{align}
\end{assumption}
When $Z\in\{0,1\}$ is random offer and $D\in\{0,1\}$ is program participation, this means that someone who received an offer to participate in the program is more likely to participate in the program than someone who didn't receive an offer.

\begin{lemma} \label{lemma:binary_iv}
Suppose Assumption \ref{assumption:binary_iv} holds. Then,
\begin{align}
1 & = \arg\max_{z\in\{0,1\}} \big\{ \eta(1, X, z)\cdot p(X, z) + \underline{y}\cdot(1-p(X, z))\big\}, \\
0 & = \arg\min_{z\in\{0,1\}} \big\{\bar{y}\cdot p(X,z) + \eta(0, X,z)\cdot(1-p(X,z))\big\}, \\
1 & =  \arg\min_{z\in\{0,1\}} \big\{\eta(1, X, z)\cdot p(X, z) + \bar{y}\cdot(1-p(X,z))\big\}, \\
0 &= \arg\max_{z\in\{0,1\}} \big\{\underline{y}\cdot p(X,z) + \eta(0, X,z)\cdot(1-p(X,z))\big\}.
\end{align}
\end{lemma}

\noindent Under Assumption \ref{assumption:binary_iv}, using Lemma \ref{lemma:binary_iv}, bounds in \eqref{equation:low_manski-iv_worstcase} and \eqref{equation:up_manski-iv_worstcase} are simplified as 
\begin{align} \label{equation:low_manski-iv_worstcase-simple}
\begin{split}
\beta_l & = E\big[\big((\eta(1, X, 1)\cdot p(X, 1) + \underline{y}\cdot(1-p(X, 1)) \\
& \qquad\qquad- (\bar{y}\cdot p(X,0) + \eta(0, X,0)\cdot(1-p(X,0)))\big) \cdot \theta_{10}(X) \\
& \qquad- \big((\eta(1, X, 1)\cdot p(X, 1) + \bar{y}\cdot(1-p(X,1))) \\
& \qquad\qquad-  (\underline{y}\cdot p(X,0) + \eta(0, X,0)\cdot(1-p(X,0)))\big) \cdot \theta_{01}(X)\big],
\end{split}
\end{align}
and
\begin{align} \label{equation:up_manski-iv_worstcase-simple}
\begin{split}
\beta_u & = E\big[\big( (\eta(1, X, 1)\cdot p(X, 1) + \bar{y}\cdot(1-p(X,1))) \\
& \qquad\qquad- (\underline{y}\cdot p(X,0) + \eta(0, X,0)\cdot(1-p(X,0)))\big)\cdot \theta_{10}(X) \\
& \qquad- \big( ( \eta(1, X, 1)\cdot p(X,1) + \underline{y}\cdot(1-p(X, 1))) \\
& \qquad \qquad - (\bar{y}\cdot p(X,0) + \eta(0, X,0)\cdot(1-p(X,0)))\big) \cdot \theta_{01}(X)\big].
\end{split}
\end{align}	

\noindent Bounds in \eqref{equation:low_manski-iv_mtr} and \eqref{equation:up_manski-iv_mtr} can also be simplified similarly.

\subsubsection{Mean monotonicity}
\noindent Next, I consider monotone instrumental variable (MIV) assumption introduced by \cite{manski2000monotone} which weakens Assumption \ref{assumption:iv} by replacing the equality in \eqref{equation:iv} by an inequality. An instrumental variable which satisfies this assumption could also help us obtain tighter bounds.
\begin{assumption}[MIV Assumption] \label{assumption:miv}	
	There exists an instrumental variable $Z\in\mathcal Z \subset \mathbb R^{d_z}$ such that, for $d\in\{0,1\}$, the following mean monotonicity holds:
	\begin{equation}
	E[Y_d|X,Z=z]\geq E[Y_d|X,Z=z'],	
	\end{equation}
	for all $z,z' \in \mathcal Z$ such that $z\geq z'$.
\end{assumption}

In the job training program example, the pre-program earnings can be used as an monotone instrumental variable when the outcome variable is post-program earnings. 

Suppose Assumption \ref{assumption:miv} holds. Then, the sharp identified region of the welfare gain is given by 
\begin{align} \label{equation:id_miv}
\begin{split}
	B_I(\delta, \delta^*) \equiv &
	\{\beta \in \mathbb{R}: \beta = E\big[E[Y_1-Y_0|X]\cdot(\delta(X)- \delta^*(X))\big], (Y_1, Y_0) \in \mathcal S(\mathcal{Y}_1 \times \mathcal{Y}_0), \\
	&\qquad \qquad (Y_1, Y_0) \text{ satisfies Assumption \ref{assumption:miv}}\}.
\end{split}	
\end{align} 
\begin{lemma}[MIV]\label{lemma:expectedsupport_miv}
	Let $(\mathcal{Y}_1 \times \mathcal{Y}_0): \Omega \to \mathcal F$ be an integrable random set that is almost surely convex and let $(Y_1,Y_0)\in \mathcal S(\mathcal{Y}_1 \times \mathcal{Y}_0)$. Let $v_1=(1,0)'$ and $v_0=(0,1)'$. Suppose Assumption \ref{assumption:miv} holds. Then, we have
	\begin{equation}  \label{equation:miv_delta}
	\begin{split}
	&\sum_{z\in \mathcal Z}P(Z=z)\cdot\big(\sup_{z_1 \leq z} \big\{ -E[s(-v_1,\mathcal{Y}_1 \times \mathcal{Y}_0)|X, Z=z_1]\big\}- \inf_{z_2 \geq z} \big\{ E[s(v_0,\mathcal{Y}_1 \times \mathcal{Y}_0)|X, Z=z_2]\big\}\big) \\
	& \qquad \qquad\qquad\qquad\qquad\qquad\qquad\qquad \leq E[Y_1-Y_0|X] \leq \\
	& \sum_{z\in \mathcal Z}P(Z=z)\cdot\big( \inf_{z_2 \geq z} \big\{  E[s(v_1,\mathcal{Y}_1 \times \mathcal{Y}_0)|X, Z=z_2] \big\}- \sup_{z_1 \leq z}\big\{ -E[s(-v_0,\mathcal{Y}_1 \times \mathcal{Y}_0)|X, Z=z_1]\big\}\big)\indent a.s.
	\end{split}
	\end{equation}
\end{lemma}

\begin{theorem}[MIV] \label{theorem:miv}
Suppose $(\mathcal{Y}_1 \times \mathcal{Y}_0): \Omega \to \mathcal F$ is an integrable random set that is almost surely convex.  Let $\delta: \mathcal{X} \rightarrow \{0,1\}$ and $\delta^*: \mathcal{X} \rightarrow \{0,1\}$ be treatment rules. Also, let $v_1=(1,0)'$ and $v_0=(0,1)'$. Then, $B_I(\delta, \delta^*)$ in \eqref{equation:id_miv} is an interval $[\beta_l, \beta_u]$ where
\begin{equation} \label{equation:low_miv}
	\beta_l = E\big[\underline{\Delta}(X)\cdot \theta_{10}(X) - \bar{\Delta}(X) \cdot \theta_{01}(X)\big],
\end{equation}
and
\begin{equation} \label{equation:up_miv}
	\beta_u = E\big[\bar{\Delta}(X) \cdot \theta_{10}(X) -\underline{\Delta}(X)\cdot \theta_{01}(X)\big],
\end{equation}
where 
\begin{equation*}
\underline{\Delta}(X)\equiv \sum_{z\in \mathcal Z}P(Z=z)\cdot\big(\sup_{z_1 \leq z}  \big\{ -E[s(-v_1,\mathcal{Y}_1 \times \mathcal{Y}_0)|X, Z=z_1]  \big\}- \inf_{z_2 \geq z} \big\{E[s(v_0,\mathcal{Y}_1 \times \mathcal{Y}_0)|X, Z=z_2] \big\}\big),
\end{equation*}	
and
\begin{equation*}
\bar{\Delta}(X) \equiv \sum_{z\in \mathcal Z}P(Z=z) \cdot\big( \inf_{z_2 \geq z} \big\{ E[s(v_1,\mathcal{Y}_1 \times \mathcal{Y}_0)|X, Z=z_2] \big\} - \sup_{z_1 \leq z}  \big\{ -E[s(-v_0,\mathcal{Y}_1 \times \mathcal{Y}_0)|X, Z=z_1] \big\} \big).
\end{equation*}	
\end{theorem}

\begin{corollary}[MIV-worst case]\label{corollary:miv_worstcase}
Let $\delta: \mathcal{X} \rightarrow \{0,1\}$ and $\delta^*: \mathcal{X} \rightarrow \{0,1\}$ be treatment rules. Then,  $B_I(\delta, \delta^*)$ in \eqref{equation:id_miv} is an interval $[\beta_l, \beta_u]$ where
\begin{align}
\begin{split}
	\beta_l & = E\big[\sum_{z\in \mathcal Z}P(Z=z)\cdot \big(\sup_{z_1 \leq z} \big\{ \eta(1, X, z_1)\cdot p(X, z_1) + \underline{y}\cdot(1-p(X, z_1))\big\} \\
	& \qquad \qquad\qquad\qquad\qquad - \inf_{z_2 \geq z} \big\{\bar{y}\cdot p(X,z_2) + \eta(0, X,z_2)\cdot(1-p(X,z_2))\big\}\big) \cdot \theta_{10}(X) \\
	& \qquad - \sum_{z\in \mathcal Z}P(Z=z) \cdot \big( \inf_{z_2 \geq z} \big\{\eta(1, X, z_2)\cdot p(X, z_2) + \bar{y}\cdot(1-p(X,z_2))\big\} \\
	&\qquad \qquad\qquad\qquad\qquad - \sup_{z_1 \leq z} \big\{\underline{y}\cdot p(X,z_1) + \eta(0, X,z_1)\cdot(1-p(X,z_1))\big\}\big) \cdot \theta_{01}(X)\Big],
\end{split}
\end{align}
and
\begin{align}
\begin{split}
	\beta_u & = E\big[\sum_{z\in \mathcal Z}P(Z=z)\cdot\big( \inf_{z_2 \geq z} \big\{ \eta(1, X, z_2)\cdot p(X,z_2) + \bar{y}\cdot (1-p(X, z_2))\big\} \\ 
	&\qquad \qquad\qquad\qquad\qquad - \sup_{z_1 \leq z} \big\{\underline{y}\cdot p(X, z_1) + \eta(0, X,z_1)\cdot(1-p(X,z_1))\big\}\big) \cdot \theta_{10}(X) \\
	& \qquad - \sum_{z\in \mathcal Z}P(Z=z) \cdot \big( \sup_{z_1 \leq z} \big\{ \eta(1, X, z_1)\cdot p(X, z_1) + \underline{y}\cdot(1-p(X,z_1))\big\}\\
	&\qquad \qquad\qquad\qquad\qquad  - \inf_{z_2 \geq z} \big\{\bar{y}\cdot p(X, z_2) + \eta(0, X, z_2)\cdot(1-p(X, z_2))\big\}\big)\cdot \theta_{01}(X)\big].
\end{split}
\end{align}
\end{corollary}

\begin{corollary}[MIV-MTR]\label{corollary:miv_mtr}
	Suppose Assumption \ref{assumption:mtr} holds. Let $\delta: \mathcal{X} \rightarrow \{0,1\}$ and $\delta^*: \mathcal{X} \rightarrow \{0,1\}$ be treatment rules. Then,  $B_I(\delta, \delta^*)$ in \eqref{equation:id_miv} is an interval $[\beta_l, \beta_u]$ where
\begin{align} \label{equation:low_miv_mtr}
\begin{split}
	\beta_l & = E\big[\sum_{z\in \mathcal Z}P(Z=z)\cdot \big(\sup_{z_1 \leq z}  \big\{ E[Y|X, Z=z_1]\big\}- \inf_{z_2 \geq z} \big\{E[Y|X, Z=z_2]\big\}\big)\cdot \theta_{10}(X) \\
	& \qquad -\sum_{z\in \mathcal Z}P(Z=z)\cdot \big( \inf_{z_2 \geq z} \big\{\eta(1, X, z_2)\cdot p(X, z_2) + \bar y\cdot (1-p(X, z_2))\big\} \\
	&\qquad \qquad \qquad\qquad - \sup_{z_1 \leq z} \big\{ \underline{y}\cdot p(X, z_1) + \eta(0, X,z_1)\cdot(1-p(X,z_1))\big\}\big) \cdot \theta_{01}(X)\big],
\end{split}
\end{align}
and	
\begin{align} \label{equation:up_miv_mtr}
\begin{split}
\beta_u & = E\big[\sum_{z\in \mathcal Z}P(Z=z)\cdot\big(\inf_{z_2 \geq z} \big\{\eta(1, X, z_2)\cdot p(X, z_2) + \bar y\cdot(1-P(X,z_2))\big\} \\
& \qquad \qquad \qquad \qquad- \sup_{z_1 \leq z}\big\{ \underline{y}\cdot p(X, z_1) + \eta(0, X,z_1)\cdot(1-p(X,z_1))\big\}\big) \cdot \theta_{10}(X) \\
& \qquad - \sum_{z\in \mathcal Z}P(Z=z)\cdot\big(\sup_{z_1 \leq z} \big\{ E[Y|X, Z=z_1] \big\}- \inf_{z_2 \geq z} \big\{  E[Y|X, Z=z_2]\big\}\big) \cdot \theta_{01}(X)\big].
\end{split}
\end{align}
\end{corollary}

Table \ref{table:forms} summarizes the forms of lower and upper bounds on CATE under different sets of assumptions.

\begin{table}[!htbp] 
	\centering
	\caption{}
	\label{table:forms}
	\begin{threeparttable}
	\fontsize{9}{9}
		\begin{tabular}{lll}	
			\hline \\[-1.8ex] 
		$\text{	Assumptions}$ & $\underline{\Delta}(X)$ & $\bar{\Delta}(X)$ \\
			\hline
			&&\\
		$\text{	worst-case}$ & $ (\eta(1,X)-\bar{y})\cdot p(X)$ &  $(\eta(1,X)-\underline{y})\cdot p(X)$ \\
			
		& $\quad + (\underline{y}-\eta(0,X))\cdot(1-p(X))$ & $\quad + (\bar{y}-\eta(0,X))\cdot(1-p(X))$ \\	
			&&\\
			&&\\
		$\text{	MTR }$ & $0$ & $\text{same as worst-case}$ \\
			
			&& \\
			&&\\
$\text{	IV-worst-case}$ & $\sup_{z\in\mathcal Z} \big\{ \eta(1, X, z)\cdot p(X, z) $  &  $\inf_{z\in\mathcal Z} \big\{\eta(1, X, z)\cdot p(X, z) $ \\
	& $\qquad\qquad + \underline{y}\cdot(1-p(X, z))\big\}$ & $\qquad\qquad + \bar{y}\cdot(1-p(X,z))\big\}$ \\
		& $- \inf_{z\in\mathcal Z} \big\{\bar{y}\cdot p(X,z)  $ & $- \sup_{z\in\mathcal Z} \big\{\underline{y}\cdot p(X,z) $ \\
			
		& $\qquad\qquad+\eta(0, X,z)\cdot(1-p(X,z))\big\}\big)$ & $\qquad\qquad+ \eta(0, X,z)\cdot(1-p(X,z))\big\}\big)$ \\
			&& \\
			&&\\
		$\text{	IV-MTR}$ & $0$ &  $\text{same as IV-worst-case} $ \\
			&& \\
			&&\\
	$\text{	MIV-worst-case}$ & $\sum_{z\in \mathcal Z}P(Z=z)\cdot \big(\sup_{z_1 \leq z} \big\{ \eta(1, X, z_1)\cdot p(X, z_1) $ & $\sum_{z\in \mathcal Z}P(Z=z) \cdot \big( \inf_{z_2 \geq z} \big\{\eta(1, X, z_2)\cdot p(X, z_2) $ \\
			
		& $\qquad\qquad\qquad\qquad\qquad + \underline{y}\cdot(1-p(X, z_1))\big\}$ & $\qquad\qquad\qquad\qquad\qquad + \bar{y}\cdot(1-p(X,z_2))\big\}$ \\

		& $\qquad\qquad\qquad-\inf_{z_2 \geq z} \big\{\bar{y}\cdot p(X,z_2)$ & $\qquad\qquad\qquad -\sup_{z_1 \leq z} \big\{\underline{y}\cdot p(X,z_1)$ \\
			
		& $\qquad\qquad\qquad\qquad + \eta(0, X,z_2)\cdot(1-p(X,z_2))\big\}\big)$ &  $\qquad\qquad\qquad\qquad+ \eta(0, X,z_1)\cdot(1-p(X,z_1))\big\}\big)$ \\
			&& \\
			&&\\
		
		$\text{	MIV-MTR}$ & $\sum_{z\in \mathcal Z}P(Z=z)\cdot \big(\sup_{z_1 \leq z}  \big\{ E[Y|X, Z=z_1]\big\}$ & $\text{same as MIV-worst-case} $ \\
		& $\qquad\qquad\qquad - \inf_{z_2 \geq z} \big\{E[Y|X, Z=z_2]\big\}\big)$  &\\
			&&\\
		\hline
		\end{tabular}
		\begin{tablenotes}
		\item \footnotesize This table reports the form of $\underline{\Delta}(X)$ and $\bar{\Delta}(X)$ under different assumptions. 
		\end{tablenotes}
	\end{threeparttable}
\end{table}

\section{Estimation and Inference} \label{section:estimation_and_inference}
The bounds developed in Section \ref{section:identification} are functions of conditional mean treatment responses $\eta(1,x)$ and $\eta(0,x)$, and propensity score $p(x)$ in the absence of instruments. The bounds with instruments are functions of conditional mean treatment responses $\eta(1,x,z)$ and $\eta(0,x,z)$, and propensity score $p(x,z)$.  Let $F$ be the joint distribution of $W=(Y,D,X,Z)$ and suppose we have a size $n$ random sample $\{w_i\}_{i=1}^{n}$ from $W$.

If the conditioning variables $X$ and $Z$ are discrete and take finitely many values, conditional mean treatment responses and propensity scores can be estimated by the corresponding empirical means. If there is a continuous component, conditional mean treatment responses and propensity scores can be estimated using nonparametric regression methods. I start with bounds that do not rely on instruments. Let $\hat\eta(1,x)$, $\hat\eta(0,x)$, and $\hat p(x)$ be those estimated values. A natural sample analog estimator for the lower bound under the worst-case in \eqref{equation:low_worstcase} can be constructed by first plugging these estimated values into \eqref{equation:low_worstcase} and then by taking average over $i$ as follows:
\begin{align} \label{equation:low_worstcase_sample}
\begin{split}	
\hat {\beta_l} & = \frac{1}{n} \sum_{i=1}^{n}\big[\big((\hat \eta(1,x_i)-\bar{y})\cdot \hat p(x_i) + (\underline{y}-\hat\eta(0,x_i))\cdot(1-\hat p(x_i))\big)\cdot\theta_{10}(x_i)\\
&\qquad - \big((\hat\eta(1,x_i)-\underline{y})\cdot \hat p(x_i) + (\bar{y}-\hat\eta(0,x_i))\cdot(1-\hat p(x_i))\big) \cdot\theta_{01}(x_i)\big].
\end{split}
\end{align}

In this estimation problem, $\eta(1,x)$, $\eta(0,x)$, and $p(x)$ are nuisance parameters that need to be estimated nonparametrically. In what follows, I collect these possibly infinite-dimensional nuisance parameters and denote it as follows:\footnote{I use $\eta(1,\cdot), \eta(0,\cdot),$ and $ p(\cdot)$ instead of $\eta(1,x), \eta(0,x),$ and $p(x)$ to highlight the fact that they are functions.}
\begin{align}
\gamma = \Big(\eta(1,\cdot), \eta(0,\cdot), p(\cdot)\Big).
\end{align} 
 Estimation of these parameters can affect the sampling distribution of $\hat{\beta_l}$ in a complicated manner. To mitigate the effect of this first-step nonparametric estimation, one could use an orthogonalized moment condition, which I describe below, to estimate ${\beta_l}$. 

Let $\beta_*$ denote either the lower bound or the upper bound, i.e., $\beta_*\in\{\beta_l, \beta_u\}$. I write my estimator as a generalized method of moments (GMM) estimator in which the true value $\beta_{*,0}$ of $\beta_*$ satisfies a single moment restriction 
\begin{equation}
E[m(w_i,\beta_{*,0}, \gamma_0)] = 0,
\end{equation}
where 
\begin{equation}\label{equation:general-moment-low}
m(w,\beta_l, \gamma) = \underline \Delta (\gamma) \cdot \theta_{10}(x) - \bar \Delta (\gamma) \cdot \theta_{01}(x)-\beta_l,
\end{equation}
and
\begin{equation}\label{equation:general-moment-up}
m(w,\beta_u, \gamma) = \bar \Delta (\gamma) \cdot \theta_{10}(x) - \underline \Delta (\gamma) \cdot \theta_{01}(x)-\beta_u.
\end{equation}

$\underline\Delta (\gamma)$ and $\bar\Delta (\gamma)$ denote the lower and upper bound on CATE respectively and are functions of the nuisance parameters $\gamma$. 

We would like our moment function  to have an \emph{orthogonality} property so that the estimation of parameter of interest would be first-order insensitive to nonparametric estimation errors in the nuisance parameter. This allows for the use of various nonparametric estimators of these parameters including high-dimensional machine learning estimators. I construct such moment function by adding influence function adjustment term for first step estimation $\phi(w, \beta_*, \gamma)$ to the original moment function $m(w,\beta_*, \gamma)$ as in CEINR. Let the \emph{orthogonalized} moment function be denoted by
\begin{equation}
\psi(w,\beta_*,\gamma) = m(w,\beta_*, \gamma)+\phi(w, \beta_*, \gamma).
\end{equation}
Let $F_\tau = (1-\tau)F_0 + \tau G$ for $\tau\in[0,1]$, where $F_0$ is the true distribution of $W$ and $G$ is some alternative distribution. Then, we say that the moment condition satisfies the Neyman orthogonality condition or is locally robust if
\begin{equation}
\frac{d }{d\tau}E[\psi(w_i,\beta_{*,0},\gamma(F_\tau))]\bigg |_{\tau=0} = 0.
\end{equation}
The orthogonality has been used in semiparametric problems by \cite{newey1990semiparametric, newey1994asymptotic}, \cite{andrews1994asymptotics}, \cite{robins1995semiparametric}, among others. More recently, in a high-dimensional setting, it has been used by \cite{belloni2012sparse}, \cite{belloni2014inference}, \cite{farrell2015robust}, \cite{belloni2017program}, \cite{athey2018approximate}, and \cite{chernozhukov2018double}, among others. Recently, \cite{sasaki2018estimation} proposed using orthogonalized moments for the estimation and inference of a parameter called policy relevant treatment effect (PRTE) whose explanation can be found in \cite{heckman2007econometric}. Much like our problem, the estimation of the PRTE involves estimation of multiple nuisance parameters. 

\subsection{Influence function calculation}
In this subsection, I show how I derive the adjustment term $\phi(w, \beta_l, \gamma)$ for the lower bound under the worst-case assumption. This illustrates how I derive the adjustment term for the cases in which $\underline \Delta(\gamma)$ and $\bar{\Delta}(\gamma)$ are differentiable with respect to $\gamma$, i.e., cases in which we do not have instrumental variables. Additional assumptions need to be imposed for the cases where $\underline \Delta(\gamma)$ and $\bar{\Delta}(\gamma)$ are non-differentiable with respect to $\gamma$.

Under the worst-case assumption, the original moment function for lower bound takes the following form:
 \begin{equation}
 \begin{split}
  m(w,\beta_l, \gamma) = & \big((\eta(1,x)-\bar{y})\cdot p(x) + (\underline{y}-\eta(0,x))\cdot (1-p(x))\big)\cdot \theta_{10}(x) \\ 
 & -\big((\eta(1,x)-\underline{y})\cdot p(x) + (\bar{y}-\eta(0, x))\cdot(1-p(x))\big) \cdot \theta_{01}(x)-\beta_l.
  \end{split}
 \end{equation}
\begin{assumption} \label{assumption:IF}
$\eta(1,x), \eta(0,x),$ and $p(x)$ are continuous at every $x$.
\end{assumption}

\begin{lemma} \label{lemma:influencefunction}
If Assumption \ref{assumption:IF} is satisfied then the influence function of $E[m(w,\beta_{l,0}, \gamma(F))]$ is $ \phi(w, \beta_{l,0}, \gamma_0)$ which is given by 
\begin{align}
\phi(w, \beta_{l,0}, \gamma_0)  & = \phi_1 + \phi_2,
\end{align}
where
\begin{align}
\begin{split}
\phi_1 &=(\theta_{10}(x)-\theta_{01}(x))\cdot(\eta_0(1,x)+\eta_0(0,x) - (\underline y + \bar y)) \cdot (d- p_0(x)), \\
\phi_2& = (\theta_{10}(x)-\theta_{01}(x)) \cdot [y - \eta_0(1,x)]^d\cdot [-(y - \eta_0(0,x))]^{1-d}.
\end{split}
\end{align}
\end{lemma}
Note that we have $E[\phi(w, \beta_{l,0}, \gamma_0)]=0$ so that the orthogonalized moment condition $\psi(w,\beta_l,\gamma)$ still identifies our parameter of interest with  $E[\psi(w, \beta_{l,0}, \gamma_0)]=0$. The adjustment term consists of two terms. While term $\phi_1$ represents the effect of local perturbations of the distribution of $D|X$ on the moment, term $\phi_2$ represents the effect of local perturbations of the distribution of $Y|D,X$ on the moment.

\subsection{GMM estimator and its asymptotic variance}
Following CEINR, I use cross-fitting, a version of sample splitting, in the construction of sample moments. Cross-fitting works as follows. Let $K>1$ be a number of folds. Partitioning the set of observation indices $\{1,2,...,n\}$ into $K$ groups $\mathcal I_k, k=1,...,K$, let $\hat\gamma_{k}$ be the first step estimates constructed from all observations not in $\mathcal I_k$. Then, $\hat{\beta_*}$ can be obtained as a solution to 
\begin{equation}
 \frac{1}{n} \sum_{k=1}^{K}\sum_{i \in \mathcal I_k} \psi(w_i, \hat{\beta_*}, \hat{\gamma}_{k})=0.
\end{equation}
\begin{assumption} \label{assumption:ceinr_ass1}
	For each $k=1,...,K,$ $(i) \int \|\psi(w,\beta_{*,0}, \hat\gamma_{k})-\psi(w, \beta_{*,0}, \gamma_0)\|^2F_0(dw)\overset{p}{\longrightarrow} 0,$ \\
	$(ii) \|\int\psi(w,\beta_{*,0}, \hat\gamma_{k})F_0(dw)\| \leq C\|\hat\gamma_{k}-\gamma_0\|^2$ for $C>0$,
	$(iii) \|\hat\gamma_{k}-\gamma_0\| = o_p(n^{-1/4})$
	$(iv)$ there is $\zeta>0$ and $d (w_i)$ with $E[ d (w_i)^2]<\infty$ such that for $\|\beta_*-\beta_{*,0}\|$ and $\|\gamma-\gamma_0\|$ small enough
	\begin{equation*}
	|\psi(w_i, \beta_*,\gamma)- \psi(w_i, \beta_{*,0},\gamma_0)| \leq d(w_i)(\|\beta_*-\beta_{*,0}\|^{\zeta}+\|\gamma-\gamma_0\|^{\zeta}).
	\end{equation*}
\end{assumption}

\begin{theorem} \label{theorem:asymptoticvariance}
Suppose that $\{w_i\}_{i=1}^{n}$ are i.i.d., Assumption \ref{assumption:ceinr_ass1} $(i)$, $(ii)$, and $(iii)$ are satisfied, $\hat{\beta_*}	\overset{p}{\longrightarrow} \beta_{*,0}$, and $\Omega_* \equiv E[\psi(w_i, \beta_{*,0},\gamma_0)^2]<\infty$. Then
\begin{equation}
\sqrt{n} (\hat{\beta_*}-\beta_{*,0}) \overset{d}{\longrightarrow}  \mathcal N(0,\Omega_*).	
\end{equation}
Moreover, if Assumption \ref{assumption:ceinr_ass1} $(iv)$ is also satisfied, a consistent estimator for the asymptotic variance can be constructed as 
\begin{equation}
\hat \Omega_* = \frac{1}{n}\sum_{k=1}^{K} \sum_{i \in \mathcal I_k} \psi(w_i, \hat{\beta_*}, \hat\gamma_{k})^2.
\end{equation}
\end{theorem}

\begin{corollary}[Locally robust estimator of the lower bound under the worst-case and a consistent estimator of its asymptotic variance] \label{corollary:variance_worstcase}
A locally robust estimator $\hat{\beta_l}$ of the lower bound under the worst-case takes the form
\begin{align}\label{equation:beta_hat}
\begin{split}
	\hat {\beta_l} & =  \frac{1}{n} \sum_{k=1}^{K}\sum_{i \in \mathcal I_k} \Big[\big((\hat\eta_{k}(1,x_i)-\bar{y})\cdot \hat p_{k}(x_i) + (\underline{y}-\hat \eta_{k}(0,x_i))\cdot(1-\hat p_{k}(x_i))\big)\cdot\theta_{10}(x_i)  \\
	& \qquad\qquad\qquad -\big((\hat \eta_{k}(1,x_i)-\underline{y})\cdot\hat p_{k}(x_i) + (\bar{y}-\hat\eta_{k}(0, x_i))\cdot(1-\hat p_{k}(x_i))\big) \cdot\theta_{01}(x_i)\\
	& \qquad\qquad\qquad +  (\theta_{10}(x_i)-\theta_{01}(x_i))\cdot(\hat\eta_{k}(1,x_i)+\hat\eta_{k}(0,x_i) - (\underline y + \bar y)) \cdot (d_i - \hat p_{k}(x_i)) \\
	& \qquad\qquad\qquad+ (\theta_{10}(x_i)-\theta_{01}(x_i))\cdot [y_i - \hat\eta_{k}(1,x_i)]^{d_i}\cdot[-(y_i - \hat\eta_{k}(0,x_i))]^{1-d_i}\Big].
\end{split}
\end{align}
Moreover, a consistent estimator of its asymptotic variance takes the form
\begin{align} \label{equation:variance_hat}
\begin{split}
\hat \Omega_l  & =	\frac{1}{n}\sum_{k=1}^{K} \sum_{i \in \mathcal I_k} \psi(w_i, \hat\beta_l, \hat\gamma_{k})^2 \\
 & = \frac{1}{n}\sum_{k=1}^{K} \sum_{i \in \mathcal I_k} \Big[\big((\hat\eta_{k}(1,x_i)-\bar{y})\cdot\hat p_{k}(x_i) + (\underline{y}-\hat\eta_{k}(0,x_i))\cdot(1-\hat p_{k}(x_i))\big)\cdot\theta_{10}(x_i)  \\
&\qquad\qquad\qquad-\big((\hat\eta_{k}(1,x_i)-\underline{y})\cdot\hat p_{k}(x_i) + (\bar{y}-\hat\eta_{k}(0, x_i))\cdot(1-\hat p_{k}(x_i))\big)\cdot\theta_{01}(x_i)-\hat{\beta_l}\\
&\qquad\qquad\qquad+ (\theta_{10}(x_i)-\theta_{01}(x_i))\cdot(\hat\eta_{k}(1,x_i)+\hat\eta_{k}(0,x_i) - (\underline y + \bar y))\cdot (d_i- \hat p_{k}(x_i)) \\
&\qquad\qquad\qquad + (\theta_{10}(x_i)-\theta_{01}(x_i))\cdot [y_i - \hat\eta_{k}(1,x_i)]^{d_i}\cdot[-(y_i - \hat\eta_{k}(0,x_i))]^{1-d_i}\Big]^2.
\end{split}
\end{align}
\end{corollary}

\bigskip

Given locally robust estimators $\hat{\beta}_l$ and $\hat{\beta}_u$ of the lower and upper bound $\beta_l$ and $\beta_u$, and consistent estimators $\hat\Omega_l$ and $\hat\Omega_u$ of their asymptotic variance $\Omega_l$ and $\Omega_u$, we can construct the $100\cdot \alpha\%$ confidence interval for the lower bound $\beta_l$ and upper bound $\beta_u$  as 

\begin{equation}
CI_\alpha^{\beta_l} =[\hat\beta_l - C_\alpha\cdot (\hat\Omega_l/n)^{1/2},\hat\beta_l+ C_\alpha\cdot (\hat\Omega_l/n)^{1/2}],
\end{equation}
and 
\begin{equation}
CI_\alpha^{\beta_u} =[\hat\beta_u - C_\alpha\cdot (\hat\Omega_u/n)^{1/2},\hat\beta_u+ C_\alpha\cdot (\hat\Omega_u/n)^{1/2}],
\end{equation}
where $C_\alpha$ satisfies
\begin{equation} \label{equation:critical-value}
\Phi(C_\alpha)-\Phi(-C_\alpha) = \alpha.
\end{equation}

\noindent In other words, $C_\alpha$ is the value that satisfies $\Phi(C_\alpha)=(\alpha+1)/2$, i.e, the $(\alpha+1)/2$ quantile of the standard normal distribution. For example, when $\alpha=0.95$, $C_\alpha$ is $1.96$.

\subsection{Bounds with instruments}
When there are additional instrumental variables, $\underline \Delta(\gamma)$ and $\bar{\Delta}(\gamma)$ in \eqref{equation:general-moment-low} and \eqref{equation:general-moment-up} are non-differentiable with respect to $\gamma$ as they involve $\sup$ and $\inf$ operators. However, under additional monotonicity assumption, the bounds can be simplified. In this section, I derive the influence function for the IV-worst-case lower bound under the monotonicity assumption. Under monotonicity, the moment condition for the IV-worst-case lower bound is
\begin{align} 
\begin{split}
m(w,\beta_l, \gamma) & = \big( \eta(1, x, 1)\cdot p(x, 1) + \underline{y}\cdot(1-p(x, 1)) \\
& \quad \quad - \bar{y}\cdot p(x,0) -\eta(0, x,0)\cdot(1-p(x,0))\big) \cdot \theta_{10}(x) \\
& \qquad- \big( \eta(1, x, 1)\cdot p(x, 1) + \bar{y}\cdot(1-p(x,1)) \\
& \quad\quad - \underline{y}\cdot p(x,0) - \eta(0, x,0)\cdot(1-p(x,0))\big) \cdot \theta_{01}(x) - \beta_l.
\end{split}
\end{align}

\begin{lemma} \label{lemma:influencefunction-iv-wc}
	If Assumption \ref{assumption:IF} is satisfied then the influence function of $E[m(w,\beta_{l,0}, \gamma(F))]$ is $ \phi(w, \beta_{l,0}, \gamma_0)$ which is given by 
	\begin{align}
	\phi(w, \beta_{l,0}, \gamma_0)  & = \phi_1 + \phi_2,
	\end{align}
	where
	\begin{align}
	\begin{split}
	\phi_1 & = [((\eta_0(1, x, 1) -\underline{y})\cdot \theta_{10}(x) - ( \eta_0(1, x, 1)- \bar y) \cdot \theta_{01}(x)) \cdot (d-p_0(x,1))]^z  \\
	& \quad \cdot [((\eta_0(0, x, 0) -\bar{y})\cdot \theta_{10}(x) - ( \eta_0(0, x, 0)- \underline y) \cdot \theta_{01}(x)) \cdot (d-p_0(x,0))] ^{1-z} \\
	\phi_2 & = (\theta_{10}(x)-\theta_{01}(x)) \cdot (\mathbbm 1\{d=1, z=1\}\cdot(y-\eta_0(1, x, 1))\\
	& \quad \quad +\mathbbm 1\{d=0, z=0\}\cdot(-(y-\eta_0(0, x, 0)))) \cdot \\
	\end{split}
	\end{align}
\end{lemma}

Notice again that we have $E[\phi(w, \beta_{l,0}, \gamma_0)]=0$ so that the orthogonalized moment condition $\psi(w,\beta_{l},\gamma)$ still identifies our parameter of interest with  $E[\psi(w, \beta_{l,0}, \gamma_0)]=0$. The adjustment term again consists of two terms. In this case, while term $\phi_1$ represents the effect of local perturbations of the distribution of $D|X, Z$ on the moment, term $\phi_2$ represents the effect of local perturbations of the distribution of $Y|D,X,Z$ on the moment. 

\clearpage
\section{Empirical Application} \label{section:empirical_application}
In this section, I illustrate my analysis using experimental data from the National Job Training Partnership Act (JTPA) Study which was commissioned by the U.S. Department of Labor in 1986. The goal of this randomized experiment was to measure the benefits and costs of training programs funded under the JTPA of 1982. Applicants who were randomly assigned to a treatment group were allowed access to the program for 18 months while the ones assigned to a control group were excluded from receiving JTPA services in that period. The original evaluation of the program is based on data of 15,981 applicants. More detailed information about the experiment and program impact estimates can be found in \cite{bloom1997benefits}. 

I follow \cite{kitagawa2018should} and focus on  adult applicants with available data on 30-month earnings after the random assignment, years of education, and pre-program earnings.\footnote{ I downloaded the dataset that \cite{kitagawa2018should} used in their analysis from \href{https://www.econometricsociety.org/content/supplement-who-should-be-treated-empirical-welfare-maximization-methods-treatment-choice}{https://www.econometricsociety.org/content/supplement-who-should-be-treated-empirical-welfare-maximization-methods-treatment-choice}. I supplemented this dataset with that of \cite{abadie2002instrumental}, which I downloaded from \href{https://economics.mit.edu/faculty/angrist/data1/data/abangim02}{https://economics.mit.edu/faculty/angrist/data1/data/abangim02}, to obtain a variable that indicates program participation. }
Table \ref{table:sumstats} shows the summary statistics of this sample. The sample consists of 9223 observations, of which 6133 (roughly 2/3) were assigned to the treatment group, and 3090 (roughly 1/3) were assigned to the control group. The means and standard deviations of program participation, 30-month earnings, years of education, and pre-program earnings are given for the entire sample, the treatment group subsample, and the control group subsample.

Treatment variable is the job training program participation and equals 1 for individuals who actually participated in the program. Only 65\% of those who got assigned to the treatment group actually participated in the training program. I look at the joint distribution of assigned and realized treatment status in Table \ref{table:compliance} to further investigate the compliance issue. Outcome variable is 30-month earnings and is on average \$16,093 and ranges from \$0 to \$155,760 with median earnings \$11,187. In the analysis below, based on this range, I set $\underline y=\$0$ and $\bar y = \$160,000$. Treatment group assignees earned \$16,487 on average while control group assignees earned \$15,311. The \$1,176 difference between these two group averages is an estimate of the JTPA impact on earnings from an intention-to-treat perspective. Pretreatment covariates I consider are years of education and pre-program earnings. Years of education are on average 11.61 years and range from 7 to 18 years with median 12 years. Pre-program earnings are on average \$3,232 and range from \$0 to \$63,000 with median earnings \$1,600. Not surprisingly, both variables are roughly balanced by assignment status due to random assignment and large samples involved. 

Although the offer of treatment was randomly assigned, the compliance was not perfect. Table \ref{table:compliance} shows the joint distribution of assigned and realized treatment. Assigned treatment equals 1 for individuals who got offered the training program and realized treatment equals 1 for individuals who actually participated in the training. As can be seen from this table, the realized treatment is not equal to assigned treatment for roughly 23\% of the applicants. Therefore, the program participation is self-selected and likely to be correlated with potential outcomes.  Since the assumption of unconfoundedness fails to hold in this case, the treatment effects are not point identified. Although the random offer can be used as a treatment variable to point identify the intention-to-treat effect as in \cite{kitagawa2018should}, the actual program participation should be used to identify the treatment effect itself.

\bigskip

\begin{table}[!htbp]
\centering
\caption{Summary statistics}
\label{table:sumstats}
\begin{threeparttable}
	\begin{tabular}{lccc}	
		\hline \\[-1.8ex] 
		& Entire sample & Assigned to & Assigned to\\
		& & treatment & control \\
	\hline \\[-1.8ex] 
	\emph{Treatment} \\
	
	Job training & 0.44 & 0.65 & 0.01 \\
	& (0.50) & (0.48) & (0.12)\\
	& \\
	\emph{Outcome variable} \\
	30-month earnings & 16,093 & 16,487 & 15,311 \\
	& (17,071) & (17,391) & (16,392) \\
		&\\
	\emph{Pretreatment covariates} \\
	Years of education & 11.61 & 11.63 & 11.58 \\
	& (1.87) & (1.87) & (1.88)\\
	Pre-program earnings & 3,232 & 3,205 & 3,287\\
	 & (4,264) & (4,279) & (4,234) \\
	 &\\
	Number of observations & 9223 & 6133 & 3090\\
		&\\
		\hline	
	\end{tabular}
	\begin{tablenotes}
		\footnotesize
		\item This table reports the means and standard deviations (in brackets) of variables in our sample. Treatment variable is job training program participation and equals 1 for individuals who actually participated in the program. The outcome variable is 30-month earnings after the random assignment. Pretreatment covariates are years of education and pre-program annual earnings. The earnings are in US Dollars.
	\end{tablenotes}	
\end{threeparttable}		
\end{table}

\bigskip

\begin{table}[!htbp] 
	\centering
	\caption{The joint distribution of assigned and realized treatment}
	\label{table:compliance}
	\begin{threeparttable}
		
		\begin{tabular}{ccccccccc}	
			\hline \\[-1.8ex] 
			&	& \multicolumn{4}{c}{Assigned treatment}  & \\
			&	Realized treatment & & 1 & & 0 & & Total & \\
			\hline \\[-1.8ex] 
			&	1 & & 4015 & & 43  & & 4058 & \\
			&	0 & & 2118 & & 3047 & & 5165 & \\
			&	Total & &  6133 & & 3090 & & 9223 & \\
			\hline	
		\end{tabular}
		\begin{tablenotes}
			\footnotesize 
			\item This table reports the joint distribution of assigned and realized treatment in our sample. Assigned treatment equals 1 for individuals who got offered job training and realized treatment equals 1 for individuals who actually participated in the training. It shows the compliance issue in our sample. 
		\end{tablenotes}
	\end{threeparttable}
\end{table}

\bigskip

\begin{example}
Applicants were eligible for training if they faced a certain barriers to employment. This included being a high school dropout. Suppose the benchmark policy is to treat everyone with less than high school education, i.e., people who have less than or equal to 11 years of education. Now, consider implementing a new policy in which we include people with high school degree. In other words, let
\begin{align}
\delta^*& = \mathbbm1\{\text{education} \leq11\}, \\
\delta & = \mathbbm1\{\text{education} \leq12\}.
\end{align}

The estimates of lower and upper bounds on the welfare gain from this new policy under various assumptions and different instrumental variables are summarized in Table \ref{table:example1}. In this example, a random offer is used as an instrumental variable and pre-program earnings is used as a monotone instrumental variable. For the first-step estimation, I use cross-fitting with $K=2$ and estimate $\hat\eta(1,x)$, $\hat\eta(0,x)$ and $\hat p(x)$ by empirical means. Those empirical means out of whole sample are depicted in Figure \ref{fig:eta-edu} and \ref{fig:pedu} in the Appendix. Empirical means and distributions when years of education is used as $X$ and random offer is used as $Z$ are summarized in Table \ref{table:emp_education} in the Appendix.

As can be seen from Table \ref{table:example1}, the worst-case bounds cover $0$, as I explained earlier. Although we cannot rank which policy is better, we quantify the no-assumption scenario as a welfare loss of $\$31,423$ and a welfare gain of  $\$36,928$. Under the MTR assumption, the lower bound is $0$. That is because the MTR assumption states that everyone benefits from the treatment, and under the new policy, we are expanding the treated population. The upper bound under MTR is the same as the upper bound under the worst-case. When we use a random offer as an instrumental variable, the bounds are tighter than the worst-case bounds and still cover 0. However, when we use pre-program earnings as a monotone instrumental variable, the bounds do not cover 0, and it is even tighter if we impose an additional MTR assumption. Therefore, if the researcher is comfortable with the validity of the MIV assumption, she can conclude that implementing the new policy is guaranteed to improve welfare and that improvement is between $\$3,569$ and $\$36,616$. 

\end{example}

\begin{table}[!htbp] \centering 
	\caption{Welfare gains in Example 1} 
	\label{table:example1} 
	\begin{threeparttable}
	\begin{tabular}{@{\extracolsep{5pt}} lcc} 
		\hline \\[-1.8ex] 
		Assumptions & lower bound & upper bound\\
		\hline \\[-1.8ex] 
		worst-case	&  -31,423 & 36,928\\
		 & (-32,564, -30,282) & (35,699, 38,158)\\
		 &\\
		MTR	 & 0 & 36,928  \\
			 &  & (35699, 38158)\\
			 &\\
		IV-worst-case 	& -2,486  & 20,787\\
		 & (-2,774, -2,198) & (19,881, 21694)\\
		 &\\
		 
		IV-MTR	 & 0 & 20,787 \\ 
		&& (19,881, 21694) \\
		& \\
		MIV-worst-case & 3,569 & 36,616 \\
		&\\
		MIV-MTR  & 7,167  & 36,616 \\
		\hline \\[-1.8ex] 
	\end{tabular} 
	\begin{tablenotes}
	\footnotesize 
	\item This table reports the estimated welfare gains and their $95\%$ confidence intervals (in brackets) in Example 1 under various assumptions. The welfare is in terms of 30-month earnings in US Dollars.
	\end{tablenotes}
	\end{threeparttable}
\end{table} 

\begin{example}
One class of treatment rules that \cite{kitagawa2018should} considered is a class of quadrant treatment rules: 
\begin{align}
\begin{split}
 \mathscr G =\{\{x : &  s_1(\text{education}-t_1)>0 \text{ and } s_2(\text{pre-program earnings}-t_2)>0\}, \\
& s_1,s_2 \in\{-1,1\}, t_1, t_2 \in \mathbb R\}\}.
\end{split}
\end{align}
	
One's education level and pre-program earnings have to be above or below some specific thresholds to be assigned to treatment according to this treatment rule. Within this class of treatment rules, the empirical welfare maximizing treatment rule that \cite{kitagawa2018should} calculates is 
$ \mathbbm 1\{\text{education} \leq 15, \text{prior earnings} \leq \$19,670 \} .$  Let this policy be the benchmark policy and consider implementing another policy that lowers the education threshold to be $12$. In fact, that policy is another empirical welfare maximizing policy that takes into account the treatment assignment cost which is $\$774$ per assignee. I calculate the welfare difference between these two policies.
In other words, let
\begin{align}
\delta^*& =\mathbbm1\{\text{education} \leq15, \text{pre-program earnings} \leq \$19,670\},\\
\delta & =\mathbbm1\{\text{education} \leq12, \text{pre-program earnings} \leq \$19,670\}.
\end{align}
The estimation results are summarized in Table \ref{table:example2}. In this example, a random offer is used as an instrumental variable. For the first-step estimation, I use cross-fitting with $K=2$ and estimate $\hat\eta(1,x)$ and  $\hat\eta(0,x)$ by polynomial regression of degree $2$ and $\hat p(x)$ by logistic regression with polynomial of degree 2. Those estimated conditional mean treatment responses and propensity score out of whole sample are depicted in Figure \ref{fig:etaeduprevearn} and \ref{fig:peduprevearn} in the Appendix. 

As can be seen from Table \ref{table:example2}, again, the worst-case bounds cover 0. However, we quantify the no-assumption scenario as a welfare loss of $\$13,435$ and a welfare gain of $\$11,633$. Under the MTR assumption, the upper bound is 0. That is because the MTR assumption states that everyone benefits from the treatment, and under the new policy, we are shrinking the treated population. The lower bound under MTR is the same as the lower bound under the worst-case. When we use a random offer as an instrumental variable, the bounds are tighter and still cover 0 as well. Using IV assumption alone, which is a credible assumption since the offer was randomly assigned in the experiment, we quantify the difference as a welfare loss of $\$7,336$ and a welfare gain of $\$1,035$. In this case, the researcher cannot be sure whether implementing the new policy is guaranteed to worsen or improve welfare. However, if she decides that the welfare gain being at most $\$1,035$ is not high enough, she can go ahead with the first policy. 

\end{example} 

\begin{table}[!htbp] \centering 
	\caption{Welfare gains in Example 2} 
	\label{table:example2} 
	\begin{threeparttable}
	\begin{tabular}{@{\extracolsep{5pt}} lcc} 
		\hline \\[-1.8ex] 
		Assumptions & lower bound & upper bound\\
		\hline \\[-1.8ex] 
		worst-case	&  -13,435 & 11,633\\
		& (-14,361, -12,510) & (10,871, 12,394) \\
		&\\
		MTR	 & -13,435 & 0  \\
		& (-14,361, -12,510) \\
		&\\
		
		IV-worst-case& -7,336 & 1,035\\
		& (-7,911, -6,763) & (862, 1,208) \\
		&\\
		IV-MTR  & -7,336  & 0 \\ 
		& (-7,911, -6,763) & \\
		\hline \\[-1.8ex] 
	\end{tabular} 
	\begin{tablenotes}
	\footnotesize 
	\item This table reports the estimated welfare gains and their $95\%$ confidence intervals (in brackets) in Example 2 under various assumptions. The welfare is in terms of 30-month earnings in US Dollars.
\end{tablenotes}
\end{threeparttable}
\end{table} 

\clearpage
\section{Simulation Study} \label{section:simulation}
Mimicking the empirical application, I consider the following data generating process. Let $X$ be a discrete random variable with values $\{7,8,9,10,11,12,13,14,15,16,17,18\}$ and probability mass function $\{0.01, 0.06, 0.07, 0.11, 0.13, 0.43, 0.07, 0.06, 0.02, 0.02, 0.01, 0.01\}$. Conditional on $X=x$, let
\begin{align}
Z|X=x & \sim Bernoulli (2/3), \\
U|X=x, Z=z &\sim Unif[0,1] \text{ for } z\in\{0,1\}, \\
D &= \mathbbm 1\{p(X,Z) \geq U \}, \\
Y_1 |X=x, Z=z, U =u & \sim Lognormal\Big(log\frac{m_1^2(x,u)}{\sqrt{\sigma_1^2+m_1^2(x,u)}},\sqrt{log(\frac{\sigma_1^2}{m_1^2(x,u)}+1)}\Big),\\
Y_0 |X=x, Z=z, U=u & \sim Lognormal\Big(log\frac{m_0^2(x,u)}{\sqrt{\sigma_0^2+m_0^2(x,u)}},\sqrt{log(\frac{\sigma_0^2}{m_0^2(x,u)}+1)}\Big),
\end{align}
where
\begin{align}
p(x,z) & = \frac{1}{1+e^{-(-4.89+0.05\cdot x + 5\cdot z)}}, \\
m_1(x,u) & = E[Y_1|X=x, Z=z, U=u] = 5591+1027 \cdot x + 2000\cdot u, \\
m_0(x,u) & =  E[Y_0|X=x, Z=z, U=u] = -1127 + 1389 \cdot x + 1000 \cdot u,\\
\sigma^2_1 & = Var[Y_1|X=x, Z=z, U=u] = 11000^2,\\
\sigma^2_0 & = Var[Y_0|X=x, Z=z, U=u] = 11000^2.
\end{align}

In this specification, $X$ corresponds to years of education and takes values from $7$ to $18$. $Z$ corresponds to random offer and follows $Bernoulli(2/3)$ to reflect the fact that probability of being randomly assigned to the treatment group is $2/3$ irrespective of applicants' years of education. $D$ corresponds to program participation and equals $1$ whenever $p(x,z)$ exceeds the value of $U$ which is uniformly distributed on $[0,1]$. $Y_1$ and $Y_0$ are potential outcomes and observed outcome $Y=Y_1\cdot D + Y_0 \cdot (1-D)$ corresponds to 30-month post-program earnings. For $d\in\{0,1\}$, $Y_d$ conditional on $X$, $Z$, and $U$ follows a lognormal distribution whose mean is $m_d(x,u)$ and variance is $\sigma_d^2$. Under this structure, we have
\begin{align}
E[Y_d |X, Z] =E[Y_d|X] \text{ for }d\in\{0,1\}.
\end{align}
As in Example 1 in Section \ref{section:empirical_application}, consider the following pair of policies:
\begin{align}
\delta^*(x)=1\{x \leq11\} \text{ and } \delta(x)=1\{x \leq12\}.
\end{align}
Policy $\delta^*$ corresponds to treating everyone who has less than or equal to 11 years of education, and policy $\delta$ corresponds to treating everyone who has less than or equal to 12 years of education. Then, the population welfare gain is 1,236. The population worst-case bounds are (-31,191, 37,608) and IV-worst-case bounds are (-2,380, 21,227). As in Section \ref{section:empirical_application}, I set $\underline y=0$ and $\bar  y=160,000$ to calculate the bounds. More details on the calculation of these population quantities can be found in Appendix \ref{appendix:simulation}.

I focus on worst-case lower bound and report coverage probabilities and average lengths of $95\%$ confidence intervals, for samples sizes $n\in \{100,1000, 5000,10000\}$, out of 1000 Monte Carlo replications in Table \ref{table:simulation-wc-coverage}. I use empirical means in the first-step estimation of conditional mean treatment responses and propensity scores. I construct the confidence intervals using original and debiased moment conditions with and without cross-fitting. Confidence intervals constructed using original moments are invalid, and as expected, show undercoverage. However, confidence intervals obtained using debiased moment conditions show good coverage even with small sample size. I also report the results when true values of nuisance parameters are used to construct the confidence intervals. In that case, the coverage probability is around $0.95$ for both original and debiased moments, as expected. 

\begin{table}[!htbp]
	\centering
	\caption{95\% confidence interval for worst-case lower bound}
	\label{table:simulation-wc-coverage} 
	\begin{tabular}{@{\extracolsep{5pt}}ccccc} 
		\\[-1.8ex]
		\hline \\[-1.8ex] 
		& \multicolumn{2}{c}{Original moment} & \multicolumn{2}{c}{Debiased moment} \\
		\\[-1.8ex]
		Sample size & Coverage & Average length & Coverage & Average length \\
		\hline	\\[-1.8ex] 
		\multicolumn{5}{l}{when first-step is estimated with empirical means}  \\
		\multicolumn{5}{c}{without cross-fitting}  \\
		100 & 0.80 & 13976 &  0.94 & 21316\\
		1000 & 0.79 & 4454 & 0.95 & 6797 \\
		5000 & 0.78 & 1995 & 0.94 & 3045 \\
		10000 & 0.80 & 1412 & 0.96 & 2154 \\
		& \\
		\multicolumn{5}{c}{with cross-fitting $(L=2)$}  \\
		100 & 0.79 & 14180 & 0.94 & 21316 \\
		1000 & 0.78 & 4462 & 0.95 & 6797 \\
		5000 & 0.78 & 1996 & 0.94 & 3045 \\
		10000 &  0.80 & 1412 & 0.96 & 2154 \\
		
		&\\	
		\multicolumn{5}{l}{when true values of nuisance parameters are used}  \\
		100 & 0.95 & 14008 & 0.94 & 21316 \\
		1000 & 0.94 & 4449 & 0.95 & 6797 \\
		5000 & 0.95 & 1991 & 0.94 & 3045 \\
		10000 & 0.95 & 1408 & 0.96 & 2154 \\
		\hline \\[-1.8ex] 
		
		\multicolumn{5}{l}{Note: number of Monte Carlo replications is 1000}	
	\end{tabular}
\end{table}

\section{Conclusion} \label{section:conclusion} 
In this paper, I consider identification and inference of the welfare gain that results from switching from one policy to another policy. Understanding how much the welfare gain is under different assumptions on the unobservables allows policymakers to make informed decisions about how to choose between alternative treatment assignment policies. I use tools from theory of random sets to obtain the identified set of this parameter. I then employ orthogonalized moment conditions for the estimation and inference of these bounds. I illustrate the usefulness of the analysis by considering hypothetical policies with experimental data from the National JTPA study. I conduct Monte Carlo simulations to assess the finite sample performance of the estimators.

\clearpage

\appendix

\section{Random Set Theory} \label{appendix:randomsettheory}
In this appendix, I introduce some definitions and theorems from random set theory that are used throughout the paper. See \cite{molchanov2017theory} and \cite{molchanov2018random} for more detailed treatment of random set theory. Let $(\Omega,\mathfrak{A}, P)$ be a complete probability space and $\mathcal F$ be the family of closed subsets of $\mathbb R^d$.

\begin{definition}[Random closed set]
	A map $X: \Omega \to \mathcal F$ is called a random closed set if, for every compact set $K$ in $\mathbb R^d$, 
\begin{equation}
\{\omega\in\Omega: X(\omega)\cap K\neq\emptyset\} \in \mathfrak A.
\end{equation}	
\end{definition}

\begin{definition}[Selection]
A random vector $\xi$ with values in $\mathbb R^d$ is called a (measurable) selection of $X$ if $\xi(\omega)\in X(\omega)$ for almost all $\omega \in \Omega$. The family of all selections of $X$ is denoted by $\mathcal S(X)$. 
\end{definition}

\begin{definition}[Integrable selection]
Let $L^1 = L^1 (\Omega; \mathbb R^d)$ denote the space of $\mathfrak{A}$-measurable random vectors with values in $\mathbb R^d$ such that the $L^1$-norm $\|\xi\|_1 =E[\|\xi\|]$ is finite. If $X$ is a random closed set in $\mathbb R^d$, then the family of all integrable selections of $X$ is given by
\begin{equation}
\mathcal S^1(X) = \mathcal S(X) \cap L^1.
\end{equation}
\end{definition}

\begin{definition}[Integrable random sets]
	A random closed set $X$ is called integrable if $\mathcal S^1(X) \neq \emptyset$.
\end{definition}

\begin{definition}[Selection (or Aumann) expectation]
	The selection (or Aumann) expectation of $X$ is the closure of the set of all expectations of integrable selections, i.e.
	\begin{equation}
	\mathbb{E}[X]=cl\{\int_{\Omega} \xi d P: \xi \in \mathcal S^1(X)\}.
	\end{equation}
\end{definition}
\noindent Note that I use $\mathbb E[\cdot]$ for the Aumann expectation and reserve $E[\cdot]$ for the expectation of random variables and random vectors. 

\begin{definition}[Support function]
Let  $K \subset \mathbb R^d$ be a convex set. The support function of a set $K$ is given by
\begin{align}
s(v,K) = \sup_{x \in K}\langle v,x \rangle,  \indent v\in \mathbb R^d.
\end{align}
\end{definition}

\begin{theorem}[Theorem 3.11 in \cite{molchanov2018random}]
	\label{theorem:support}
	If an integrable random set $X$ is defined on a nonatomic probability space, or if $X$ is almost surely convex, then
	\begin{equation}
		E[s(v,X)] =s(v, \mathbb{E}[X]), \indent v\in \mathbb{R}^d. 
	\end{equation}
\end{theorem}

\section{Proofs and Useful Lemmas}
\subsection*{Proof of Lemma \ref{lemma:expectedsupport}}
	By the definition of selection expectation, we have $E[(Y_1, Y_0)'|X] \in \mathbb{E}[\mathcal{Y}_1 \times \mathcal{Y}_0|X]$. Then by the definition of support function and Theorem \ref{theorem:support}, for any $v\in \mathbb R^2$, we have
	\begin{equation} 
	\begin{split}
	v'E[(Y_1, Y_0)'|X] &\leq s(v,\mathbb{E}[\mathcal{Y}_1 \times \mathcal{Y}_0|X])\\
	& = E[s(v,\mathcal{Y}_1 \times \mathcal{Y}_0)|X].
	\end{split}
	\end{equation}
	For any $v\in \mathbb R^2$, we can write
	\begin{equation} 
	\begin{split}
	-v'E[(Y_1, Y_0)'|X] &\leq s(-v,\mathbb{E}[\mathcal{Y}_1 \times \mathcal{Y}_0|X])\\
	& = E[s(-v,\mathcal{Y}_1 \times \mathcal{Y}_0)|X].
	\end{split}
	\end{equation}
	Thus, we also have 
	\begin{equation} 
	\begin{split}
	v'E[(Y_1, Y_0)'|X] \geq -E[s(-v,\mathcal{Y}_1 \times \mathcal{Y}_0)|X]. 
	\end{split}
	\end{equation} \qed

\subsection*{Proof of Theorem \ref{theorem:main}}
We write $\Delta(X) \equiv E[Y_1-Y_0|X] = {v^*}'E[(Y_1, Y_0)'|X]$ for $v^*=(1,-1)'$. By Lemma \ref{lemma:expectedsupport}, we have 
	\begin{equation} \label{equation:delta}
	\underline{\Delta}(X) = -E[s(-v^*,\mathcal{Y}_1 \times \mathcal{Y}_0)|X] \leq \Delta(X) \leq E[s(v^*,\mathcal{Y}_1 \times \mathcal{Y}_0)|X] = \bar{\Delta}(X) \indent a.s.
	\end{equation}
	
	\noindent Since $\delta(X)-\delta^*(X)$ can take values in $\{-1,0,1\}$, we consider two cases: (i) $\delta(X)-\delta^*(X)=1$ and (ii) $\delta(X)-\delta^*(X)=-1$. When (i) $\delta(X)-\delta^*(X)=1$, the upper bound on $\Delta(X)\cdot(\delta(X)- \delta^*(X))$ is $\bar\Delta(X)$. When (ii) $\delta(X)-\delta^*(X)=-1$, the upper bound on $\Delta(X)\cdot(\delta(X)- \delta^*(X))$ is $-\underline\Delta(X)$. Hence, the upper bound on $E[\Delta(X)\cdot(\delta(X)- \delta^*(X))]$ should be
	\begin{align}
		\beta_u = E[\bar{\Delta}(X)\cdot \theta_{10}(X) -\underline{\Delta}(X) \cdot \theta_{01}(X)].
	\end{align}
	Similarly, the lower bound on $E[\Delta(X)\cdot(\delta(X)- \delta^*(X))]$ should be
\begin{align}
\beta_l = E[\underline{\Delta}(X) \cdot\theta_{10}(X)  - \bar{\Delta}(X)\cdot \theta_{01}(X)]. 
\end{align}\qed


\begin{lemma}\label{lemma:randomset}
Suppose $(\mathcal{Y}_1 \times \mathcal{Y}_0): \Omega \to \mathcal F$ is of the following form: 
\begin{align}
\mathcal{Y}_1 \times \mathcal{Y}_0 = 
\begin{cases}
\{Y\}\times [Y_{L,0}, Y_{U,0}] $ if $ D=1, \\
[Y_{L,1}, Y_{U,1}] \times \{Y\} $ if $ D=0,
\end{cases}
\end{align}
where $Y$ is a random variable and each of $Y_{L,0},   Y_{U,0}, Y_{L,1},$ and $Y_{U,1}$ can be a constant or a random variable. Let $v^*=(1,-1)'$, $v_1=(1,0)'$, and $v_0=(0,1)'$. Then, we have
\begin{align*}
E[s(v_1, \mathcal{Y}_1 \times \mathcal{Y}_0)|X]
& = E[Y|D=1, X]\cdot P(D=1|X) + E[Y_{U,1}|D=0, X]\cdot P(D=0|X), \\
-E[s(-v_1, \mathcal{Y}_1 \times \mathcal{Y}_0)|X] 
& = E[Y|D=1, X]\cdot P(D=1|X) + E[Y_{L,1}|D=0, X]\cdot P(D=0|X), \\
E[s(v_0, \mathcal{Y}_1 \times \mathcal{Y}_0)|X] 
& = E[Y_{U,0}|D=1, X]\cdot P(D=1|X) + E[Y|D=0,X]\cdot P(D=0|X), \\
-E[s(-v_0, \mathcal{Y}_1 \times \mathcal{Y}_0)|X] 
& = E[Y_{L,0}|D=1, X]\cdot P(D=1|X) + E[Y|D=0,X]\cdot P(D=0|X), \\
E[s(v^*, \mathcal{Y}_1 \times \mathcal{Y}_0)|X]
& = (E[Y|D=1, X]-E[Y_{L,0}|D=1, X])\cdot P(D=1|X) \\
& \qquad +(E[Y_{U,1}|D=0, X]-E[Y|D=0, X])\cdot P(D=0|X), \\
-E[s(-v^*, \mathcal{Y}_1 \times \mathcal{Y}_0)|X] 
& = (E[Y|D=1, X]-E[Y_{U,0}|D=1,X])\cdot P(D=1|X) \\
& \qquad + (E[Y_{L,1}|D=0,X]-E[Y|D=0, X])\cdot P(D=0|X). 
\end{align*}
\end{lemma}

\begin{proof}
	We have	
	\begin{align*}
	\begin{split}
	E[s(v_1, \mathcal{Y}_1 \times \mathcal{Y}_0)|X] &= E[\sup_{(y_1, y_0) \in \begin{cases}
		\{Y\}\times [Y_{L,0}, Y_{U,0}] $ if $ D=1, \\
		[Y_{L,1}, Y_{U,1}] \times \{Y\} $ if $ D=0.
		\end{cases}} y_1|X]\\
	& = E[Y|D=1, X]\cdot P(D=1|X) + E[Y_{U,1}|D=0, X]\cdot P(D=0|X), \\
	& \\
	-E[s(-v_1, \mathcal{Y}_1 \times \mathcal{Y}_0)|X]&= -E[\sup_{(y_1, y_0) \in \begin{cases}
		\{Y\}\times [Y_{L,0}, Y_{U,0}] $ if $ D=1, \\
		[Y_{L,1}, Y_{U,1}] \times \{Y\} $ if $ D=0.
		\end{cases}} -y_1|X]\\
	&= E[\inf_{(y_1, y_0) \in \begin{cases}
		\{Y\}\times [Y_{L,0}, Y_{U,0}] $ if $ D=1, \\
		[Y_{L,1}, Y_{U,1}] \times \{Y\} $ if $ D=0.
		\end{cases}} y_1|X]\\
	& = E[Y|D=1, X]\cdot P(D=1|X) + E[Y_{L,1}|D=0, X]\cdot P(D=0|X),\\
	& \\
	E[s(v_0, \mathcal{Y}_1 \times \mathcal{Y}_0)|X] &= E[\sup_{(y_1, y_0) \in \begin{cases}
		\{Y\}\times [Y_{L,0}, Y_{U,0}] $ if $ D=1, \\
		[Y_{L,1}, Y_{U,1}] \times \{Y\} $ if $ D=0.
		\end{cases}} y_0|X]\\
	& = E[Y_{U,0}|D=1, X]\cdot P(D=1|X) + E[Y|D=0, X]\cdot P(D=0|X),\\
	& \\
	-E[s(-v_0, \mathcal{Y}_1 \times \mathcal{Y}_0)|X] &= -E[\sup_{(y_1, y_0) \in \begin{cases}
		\{Y\}\times [Y_{L,0}, Y_{U,0}] $ if $ D=1, \\
		[Y_{L,1}, Y_{U,1}] \times \{Y\} $ if $ D=0.
		\end{cases}} -y_0|X]\\
	&= E[\inf_{(y_1, y_0) \in \begin{cases}
		\{Y\}\times [Y_{L,0}, Y_{U,0}] $ if $ D=1, \\
		[Y_{L,1}, Y_{U,1}] \times \{Y\} $ if $ D=0.
		\end{cases}} y_0|X]\\
	& = E[Y_{L,0}|D=1, X]\cdot P(D=1|X) + E[Y|D=0, X]\cdot P(D=0|X),\\
	\end{split}	
	\end{align*}
	\begin{align*}
	\begin{split}
	E[s(v^*, \mathcal{Y}_1 \times \mathcal{Y}_0)|X] & = E[\sup_{(y_1, y_0) \in \begin{cases}
		\{Y\}\times [Y_{L,0}, Y_{U,0}] $ if $ D=1, \\
		[Y_{L,1}, Y_{U,1}] \times \{Y\} $ if $ D=0.
		\end{cases}} y_1-y_0|X]\\
	& = (E[Y|D=1, X]-E[Y_{L,0}|D=1, X])\cdot P(D=1|X)  \\ 
	&\qquad +(E[Y_{U,1}|D=0, X]-E[Y|D=0, X])\cdot P(D=0|X), \\
	& \\
	-E[s(-v^*, \mathcal{Y}_1 \times \mathcal{Y}_0)|X] &= -E[\sup_{(y_1, y_0) \in \begin{cases}
		\{Y\}\times [Y_{L,0}, Y_{U,0}] $ if $ D=1, \\
		[Y_{L,1}, Y_{U,1}] \times \{Y\} $ if $ D=0.
		\end{cases}} -y_1+y_0|X]\\
	&= E[\inf_{(y_1, y_0) \in \begin{cases}
		\{Y\}\times [Y_{L,0}, Y_{U,0}] $ if $ D=1, \\
		[Y_{L,1}, Y_{U,1}] \times \{Y\} $ if $ D=0.
		\end{cases}} y_1-y_0|X]\\
	& = (E[Y|D=1, X]-E[Y_{U,0}|D=1,X])\cdot P(D=1|X) \\
	& \qquad + (E[Y_{L,1}|D=0,X]-E[Y|D=0, X])\cdot P(D=0|X). \qedhere
	\end{split}	
	\end{align*}	
\end{proof}	

\subsection*{Proof of Corollary \ref{corollary:worstcase}}
By setting $Y_{L,1}=Y_{L,0}=\underline y$ and $Y_{U,1}=Y_{U,0}=\bar y$ in Lemma \ref{lemma:randomset}, we have
\begin{align}
E[s(v^*, \mathcal{Y}_1 \times \mathcal{Y}_0)|X]
& = (\eta(1,X)-\underline{y})\cdot p(X) + (\bar{y}-\eta(0,X))\cdot(1-p(X)),\\
-E[s(-v^*, \mathcal{Y}_1 \times \mathcal{Y}_0)|X]
& = (\eta(1,X)-\bar{y})\cdot p(X) + (\underline{y}-\eta(0,X))\cdot (1-p(X)).
\end{align}
Plugging these in, the result follows from Theorem \ref{theorem:main}. \placeqed

\subsection*{Proof of Corollary \ref{corollary:mtr}}
By setting $Y_{L,0}=\underline y$, $Y_{U,0}=Y_{L,1}=Y$, and $Y_{U,1}=\bar y$ in Lemma \ref{lemma:randomset}, we have
\begin{align}
E[s(v^*, \mathcal{Y}_1 \times \mathcal{Y}_0)|X]
& = (\eta(1,X)-\underline{y})\cdot p(X) + (\bar{y}-\eta(0,X))\cdot (1-p(X)),\\
-E[s(-v^*, \mathcal{Y}_1 \times \mathcal{Y}_0)|X] 
& = 0. 
\end{align}
Plugging these in, the result follows from Theorem \ref{theorem:main}. \placeqed

\subsection*{Proof of Lemma \ref{lemma:expectedsupport_iv}}
	By the definition of selection expectation, we have $E[(Y_1, Y_0)'|X, Z] \in \mathbb{E}[\mathcal{Y}_1 \times \mathcal{Y}_0|X, Z]$. 
	By arguments that appear in Lemma \ref{lemma:expectedsupport}, for any $v\in \mathbb R^2$ and for all $z\in\mathcal Z$, we have 
	\begin{equation} 
	-E[s(-v,\mathcal{Y}_1 \times \mathcal{Y}_0)|X, Z=z] \leq v'E[(Y_1, Y_0)'|X, Z=z] \leq E[s(v,\mathcal{Y}_1 \times \mathcal{Y}_0)|X, Z=z].
	\end{equation}
	Assumption \ref{assumption:iv} implies that 
	\begin{equation} \label{equation:iv_independence}
	E[Y_d|X,Z]=E[Y_d|X],\indent d=0,1.
	\end{equation}
	Hence, for all $z\in\mathcal Z$, the following holds:
	\begin{equation} 
	-E[s(-v,\mathcal{Y}_1 \times \mathcal{Y}_0)|X, Z=z] \leq v'E[(Y_1, Y_0)'|X] \leq E[s(v,\mathcal{Y}_1 \times \mathcal{Y}_0)|X, Z=z].
	\end{equation}
	We therefore have
	\begin{equation} 
	\sup_{z\in\mathcal Z}\big\{ -E[s(-v,\mathcal{Y}_1 \times \mathcal{Y}_0)|X, Z=z]\big\} \leq v'E[(Y_1, Y_0)'|X] \leq \inf_{z\in \mathcal Z} \big\{ E[s(v,\mathcal{Y}_1 \times \mathcal{Y}_0)|X, Z=z]\big\}. 
	\end{equation}\qed

\subsection*{Proof of Theorem \ref{theorem:iv}}
By Lemma \ref{lemma:expectedsupport_iv}, we have
\begin{equation} \label{equation:delta_iv}
\sup_{z\in\mathcal Z} \big\{ -E[s(-v^*,\mathcal{Y}_1 \times \mathcal{Y}_0)|X, Z=z]\big\} \leq \Delta(X) \leq \inf_{z\in \mathcal Z} \big\{ E[s(v^*,\mathcal{Y}_1 \times \mathcal{Y}_0)|X, Z=z]\big\} \indent a.s.
\end{equation}
The remaining part of the proof is the same as that of Theorem \ref{theorem:main}. \placeqed

\subsection*{Proof of Corollary \ref{corollary:iv_worstcase}}
The statements in Lemma \ref{lemma:randomset} still hold when we condition on an additional variable $Z$. Hence, by setting $Y_{L,1}=Y_{L,0}=\underline y$ and $Y_{U,1}=Y_{U,0}=\bar y$ in Lemma \ref{lemma:randomset}, we have 
\begin{align}
E[s(v^*, \mathcal{Y}_1 \times \mathcal{Y}_0)|X, Z=z]
& = (\eta(1,X,z)-\underline{y})\cdot p(X, z) + (\bar{y}-\eta(0,X,z))\cdot(1-p(X,z)),\\
-E[s(-v^*, \mathcal{Y}_1 \times \mathcal{Y}_0)|X, Z=z]
& = (\eta(1,X,z)-\bar{y})\cdot p(X,z) + (\underline{y}-\eta(0,X,z))\cdot(1-p(X,z)),
\end{align}
for all $z \in \mathcal Z$. Plugging these in, the result follows from Theorem \ref{theorem:iv}. \placeqed

\subsection*{Proof of Corollary \ref{corollary:iv_mtr}}
The statements in Lemma \ref{lemma:randomset} still hold when we condition on an additional variable $Z$. Hence, by setting $Y_{L,0}=\underline y$, $Y_{U,0}=Y_{L,1}=Y$, and $Y_{U,1}=\bar y$ in Lemma \ref{lemma:randomset}, we have
\begin{align}
E[s(v^*, \mathcal{Y}_1 \times \mathcal{Y}_0)|X,Z=z]
& = (\eta(1,X,z)-\underline{y})\cdot p(X,z) + (\bar{y}-\eta(0,X,z))\cdot(1-p(X,z)),\\
-E[s(-v^*, \mathcal{Y}_1 \times \mathcal{Y}_0)|X, Z=z] 
& = 0,
\end{align}
for all $z\in \mathcal Z$.
Plugging these in, the result follows from Theorem \ref{theorem:iv}. \placeqed

\subsection*{Proof of Lemma \ref{lemma:expectedsupport_miv}}
	By the definition of selection expectation, we have $E[(Y_1, Y_0)'|X, Z] \in \mathbb{E}[\mathcal{Y}_1 \times \mathcal{Y}_0|X, Z]$. 
	By arguments that appear in Lemma \ref{lemma:expectedsupport}, for any $v\in \mathbb R^2_{+}$ and for all $z\in\mathcal Z$, we have 
	\begin{equation} 
	-E[s(-v,\mathcal{Y}_1 \times \mathcal{Y}_0)|X, Z=z] \leq v'E[(Y_1, Y_0)'|X, Z=z] \leq E[s(v,\mathcal{Y}_1 \times \mathcal{Y}_0)|X, Z=z].
	\end{equation}
	By Assumption \ref{assumption:miv}, the following holds for all $z\in\mathcal Z$:  
	\begin{equation} 
	\sup_{z_1 \leq z} \big\{-E[s(-v,\mathcal{Y}_1 \times \mathcal{Y}_0)|X, Z=z_1]\big\} \leq v'E[(Y_1, Y_0)'|X, Z=z] \leq \inf_{z_2 \geq z}  \big\{E[s(v,\mathcal{Y}_1 \times \mathcal{Y}_0)|X, Z=z_2] \big\}.
	\end{equation}
	By replacing $v$ with $v_1=(1,0)'$ and $v_0=(0,1)'$ and integrating everything with respect to $Z$, we obtain the following:
	\begin{equation} \label{equation:miv_y1_low}
	\begin{split}
	E[Y_1|X] \geq \sum_{z\in \mathcal Z}P(Z=z)\cdot \big(\sup_{z_1 \leq z}  \big\{ -E[s(-v_1,\mathcal{Y}_1 \times \mathcal{Y}_0)|X, Z=z_1]\big\}\big), 
	\end{split}
	\end{equation}
	\begin{equation} \label{equation:miv_y1_up}
	\begin{split}
E[Y_1|X] \leq \sum_{z\in \mathcal Z}P(Z=z)\cdot \big( \inf_{z_2 \geq z} \big\{E[s(v_1,\mathcal{Y}_1 \times \mathcal{Y}_0)|X, Z=z_2]\big\}\big),
	\end{split}
	\end{equation}
		\begin{equation}  \label{equation:miv_y0_low}
	\begin{split}
E[Y_0|X] \geq	\sum_{z\in \mathcal Z}P(Z=z)\cdot \big(\sup_{z_1 \leq z} \big\{ -E[s(-v_0,\mathcal{Y}_1 \times \mathcal{Y}_0)|X, Z=z_1] \big\}\big),
	\end{split}
	\end{equation}
	\begin{equation} \label{equation:miv_y0_up}
	\begin{split}
E[Y_0|X] \leq
\sum_{z\in \mathcal Z}P(Z=z)\cdot \big(\inf_{z_2 \geq z}  \big\{E[s(v_0,\mathcal{Y}_1 \times \mathcal{Y}_0)|X, Z=z_2] \big\}\big).
	\end{split}
	\end{equation}
Then, the upper bound in (\ref{equation:miv_delta}) can be obtained by subtracting the lower bound on $E[Y_0|X]$ (\ref{equation:miv_y0_low}) from the upper bound on $E[Y_1|X]$ (\ref{equation:miv_y1_up}). Similarly, the lower bound in (\ref{equation:miv_delta}) can be obtained by subtracting the upper bound on $E[Y_0|X]$ (\ref{equation:miv_y0_up}) from the lower bound on $E[Y_1|X]$ (\ref{equation:miv_y1_low}). \placeqed

\subsection*{Proof of Theorem \ref{theorem:miv}}
Bounds on $\Delta(X)$ is derived in Lemma \ref{lemma:expectedsupport_miv}. The remaining part of the proof is the same as that of Theorem \ref{theorem:main}. \placeqed

\subsection*{Proof of Corollary \ref{corollary:miv_worstcase}}
The statements in Lemma \ref{lemma:randomset} still hold when we condition on an additional variable $Z$. Hence, by setting $Y_{L,1}=Y_{L,0}=\underline y$ and $Y_{U,1}=Y_{U,0}=\bar y$ in Lemma \ref{lemma:randomset}, for all $z\in \mathcal Z$, we have
\begin{align}
E[s(v_1, \mathcal{Y}_1 \times \mathcal{Y}_0)|X, Z=z]
& = \eta(1,X,z) \cdot p(X,z) + \bar{y}\cdot (1-p(X,z)), \\
-E[s(-v_1, \mathcal{Y}_1 \times \mathcal{Y}_0)|X, Z=z]
& = \eta(1,X,z)\cdot p(X,z) + \underline{y}\cdot (1-p(X,z)), \\
E[s(v_0, \mathcal{Y}_1 \times \mathcal{Y}_0)|X, Z=z]
& = \bar{y}\cdot p(X,z) + \eta(0,X,z)\cdot (1-p(X,z)), \\
-E[s(-v_0, \mathcal{Y}_1 \times \mathcal{Y}_0)|X, Z=z]
& = \underline{y}\cdot p(X,z) + \eta(0,X,z)\cdot (1-p(X,z)).
\end{align}
Plugging these in, the result follows from Theorem \ref{theorem:miv}. \placeqed

\subsection*{Proof of Corollary \ref{corollary:miv_mtr}}
The statements in Lemma \ref{lemma:randomset} still hold when we condition on an additional variable $Z$.  Hence, by setting $Y_{L,0}=\underline y$, $Y_{U,0}=Y_{L,1}=Y$, and $Y_{U,1}=\bar y$ in Lemma \ref{lemma:randomset}, for all $z\in \mathcal Z$, we also have
\begin{align}
E[s(v_1, \mathcal{Y}_1 \times \mathcal{Y}_0)|X, Z=z]
& = \eta(1,X,z)\cdot p(X,z) + \bar y\cdot (1-p(X,z)), \\
-E[s(-v_1, \mathcal{Y}_1 \times \mathcal{Y}_0)|X, Z=z]
& = E[Y|X, Z=z], \\
E[s(v_0, \mathcal{Y}_1 \times \mathcal{Y}_0)|X, Z=z] 
&= E[Y|X, Z=z], \\
-E[s(-v_0, \mathcal{Y}_1 \times \mathcal{Y}_0)|X, Z=zv]
& = \underline{y}\cdot p(X,z) + \eta(0,X,z)\cdot (1-p(X,z)),
\end{align}
Plugging these in, the result follows from Theorem \ref{theorem:miv}. \placeqed

\subsection*{Proof of Lemma \ref{lemma:influencefunction}}
	For $0\leq \tau \leq 1$, let 
	\begin{equation}
	F_\tau = (1-\tau)F_0 + \tau G_w^j,
	\end{equation} 
	where $F_0$ is the true distribution of $F$ and $G_w^j$ is a family of distributions approaching the CDF of a constant $w$ as $j\to \infty$. Let $F_0$ be absolutely continuous with pdf $f_0(w)=f_0(y,d,x)$. Let the marginal, conditional, and joint distributions and densities under $F_0$ be denoted by $F_0(x), F_0(d|x),F_0(y|d, x), F_0(d,x)$ and $f_0(x), f_0(d|x), f_0(y|d,x), f_0(d,x)$, etc. and the expectations under $F_0$ be denoted by $E_0$. As in \cite{ichimura2017influence}, let 
	\begin{equation}
	G^j_w(\tilde{w}) =E[1\{w_i\leq \tilde{w}\}\varphi(w_i)],
	\end{equation}
	where $\varphi(w_i)$ is a bounded function with $E[\varphi(w_i)] =1$. This $G^j_w(\tilde{w})$ will approach the cdf of the constant $\tilde{w}$ as $\varphi(w)f_0(w)$ approaches a spike at $\tilde w$. For small enough $\tau$, $F_\tau$ will be a cdf with pdf $f_\tau$ that is given by 
	\begin{equation}
	f_\tau(\tilde w) = f_0(\tilde w)[1-\tau + \tau \varphi (w)] = f_0(\tilde w)(1+\tau S(w)), S(w)=\varphi(w)-1.
	\end{equation}
	Let the marginal, conditional, and joint distributions and densities under $F_\tau$ be similarly denoted by $F_\tau(x), F_\tau(d|x),F_\tau(y|d, x), F_\tau(d,x)$ and $f_\tau(x), f_\tau(d|x), f_\tau(y|d,x), f_\tau(d,x)$, etc. and the expectations under $F_\tau$ be denoted by $E_\tau$. 
	By \cite{ichimura2017influence}'s Lemma A1, we have 
	
	\begin{equation} \label{equation:derivative_tauydx}
	\frac{d }{d\tau} E_\tau[Y|D=d,X=x] = E_0[\{Y-E_0[Y|D=d, X=x]\}\varphi(W)|D=d, X=x]
	\end{equation}
	and
	\begin{equation}\label{equation:derivative_taudx}
	\frac{d }{d\tau} E_\tau[ 1\{D=d\}|X=x] = E_0[\{1\{D=d\}-E_0[1\{D=d\}|X=x]\}\varphi(W)|X=x].
	\end{equation}
	The influence function can be calculated as 
	\begin{equation}
	\phi(w, \beta, \gamma) = \lim_{j \to \infty} \Big[\frac{d }{d\tau}E_\tau[m(w_i,\beta,\gamma(F_\tau))]\bigg |_{\tau=0}\Big].
	\end{equation}
	\noindent We first denote the conditional mean treatment response and the propensity score under $F_\tau$ by
	\begin{equation}
	\eta_\tau(d,x) \equiv \int y d F_{\tau}(y|d,x),
	\end{equation}
	and
	\begin{equation}
	p_\tau(x) \equiv \int 1\{d=1\} d F_{\tau}(d|x).
	\end{equation}
	Then, by the chain rule, we have
	\begin{align*}
	\frac{d }{d\tau} E_\tau[m(w_i,\beta,\gamma(F_\tau))]
	= &\frac{d }{d\tau} E_\tau[m(w_i,\beta,\gamma(F_0))] + \frac{d }{d\tau} E_0[m(w_i,\beta,\gamma(F_\tau))] \\
	= & \frac{d}{d \tau} \Bigg[ \int\Big( ((\eta_0(1,x)-\bar{y})p_0(x) + (\underline{y}-\eta_0(0,x))(1-p_0(x)))\theta_{10}(x)  \\
	&-((\eta_0(1,x)-\underline{y})p_0(x) + (\bar{y}-\eta_0(0, x))(1-p_0(x))) \theta_{01}(x))- \beta \Big)dF_{\tau}(x)   \Bigg]\\
	&+ \frac{d}{d \tau}\Bigg[\int\Big( ((\eta_0(1,x)-\bar{y})p_\tau(x) + (\underline{y}-\eta_0(0,x))(1-p_\tau(x)))\theta_{10}(x)  \\
	&-((\eta_0(1,x)-\underline{y})p_\tau(x) + (\bar{y}-\eta_0(0, x))(1-p_\tau(x))) \theta_{01}(x))- \beta \Big)dF_{0}(x) \Bigg]\\
	&+ \frac{d}{d \tau} \Bigg[\int\Big( ((\eta_\tau(1,x)-\bar{y})p_0(x) + (\underline{y}-\eta_\tau(0,x))(1-p_0(x)))\theta_{10}(x)  \\
	&-((\eta_\tau(1,x)-\underline{y})p_0(x) + (\bar{y}-\eta_\tau(0, x))(1-p_0(x))) \theta_{01}(x))- \beta \Big)dF_{0}(x)\Bigg].
	\end{align*}
	First, we have
	\begin{align*}
	\frac{d }{d\tau} E_\tau[m(w_i,\beta,\gamma(F_0))] =  & \int \Big((\eta_0(1,x)-\bar{y})p_0(x) + (\underline{y}-\eta_0(0,x))(1-p_0(x)))\theta_{10}(x)  \\
	&-((\eta_0(1,x)-\underline{y})p_0(x) + (\bar{y}-\eta_0(0, x))(1-p_0(x))) \theta_{01}(x)\Big)dG(x)-\beta.\\
	\end{align*}
	Next, we want to find $\frac{d }{d\tau} E_0[m(w_i,\beta,\gamma(F_\tau))]$.
	In order to do that, first note that we have
	\begin{align*}
	&\frac{d}{d\tau}\int\theta(x) \eta_\tau(d,x)f_0(d|x)f_0(x)dx  \\
	&=  \int \theta(x) \frac{d}{d\tau}  \Big[\eta_\tau(d,x)\Big]f_0(d|x)f_0(x)dx \\
	& = \int \theta(x) E_0[\{Y- \eta_0(d,x)\}\varphi(W)|D=d, X=x]f_0(d|x)f_0(x)dx \\
	& = \int  \theta(x)  [\int \{y - \eta_0(d,x)\}\frac{g(y,d,x)}{f_0(y,d,x)}f_0(y|d,x)dy] f_0(d|x)f_0(x)dx \\
	& = \int  \theta(x)  [\int \{y - \eta_0(d,x)\}\frac{g(y,d,x)}{f_0(y,d,x)}\frac{f_0(y,d,x)}{f_0(d,x)}dy] \frac{f_0(d,x)}{f_0(x)}f_0(x)dx \\
	& = \int \theta(x) [\int \{y - \eta_0(d,x)\}g(y,d,x)dy]dx \\
	& = \int \theta(x) \{y - \eta_0(d,x)\}g(y,d,x)dydx.\\
	\end{align*}
	The second equality follows from equation (\ref{equation:derivative_tauydx}). The third equality follows from choosing $\varphi(w)$ to be a ratio of a sharply peaked pdf to the true density:
	\begin{equation}
	\varphi(\tilde w) = \frac{g(\tilde w)1(f_0(\tilde w)\geq  1/j)}{f_0(\tilde w)},
	\end{equation}
	where as in \cite{ichimura2017influence}, $g(w)$ is specified as follows. Letting $K(u)$ be a pdf that is symmertic around zero, has bounded support, and is continuously differentiable of all orders with bounded derivatives, we let 
	\begin{equation}
	g(\tilde w) = \prod_{l=1}^{r} \kappa_l^j (\tilde w_l), \kappa_l^j (\tilde w_l)=  \frac{jK((w_l-\tilde w_l)j)}{j \int K((w_l-\tilde w_l)j)d\mu_l(\tilde w_l)}.
	\end{equation}
	Hence, we obtain 
	\begin{align*}
	&\frac{d}{d \tau} \Bigg[\int\Big( ((\eta_\tau(1,x)-\bar{y})p_0(x) + (\underline{y}-\eta_\tau(0,x))(1-p_0(x)))\theta_{10}(x)  \\
	&-((\eta_\tau(1,x)-\underline{y})p_0(x) + (\bar{y}-\eta_\tau(0, x))(1-p_0(x))) \theta_{01}(x))- \beta \Big)dF_{0}(x)\Bigg] \\
	& = \int  (\theta_{10}(x)-\theta_{01}(x)) \{y - \eta_0(1,x)\}g(y,1,x)dydx  \\
	& - \int (\theta_{10}(x)-\theta_{01}(x)) \{y - \eta_0(0,x)\}g(y,0,x)dydx.
	\end{align*}
	With the similar argument, but using equation (\ref{equation:derivative_taudx}), we also have
	\begin{align*}
	&\frac{d}{d\tau}\int \theta(x)p_\tau(x)f_0(x)dx  \\
	& =  \int \theta(x) \frac{d}{d\tau}  \Big[p_\tau(x)\Big]f_0(x)dx \\
	& = \int \theta(x) E_0[\{1\{D=1\}-p_0(x)\}\varphi(W)|X=x] f_0(x)dx \\
	& = \int \theta(x)\Big[  \int \{ 1\{d=1\} - p_0(x)\}\frac{g(y,d,x)}{f_0(y,d,x)} f_0(y,d|x) dy dd \Big] f_0(x)dx \\
	& = \int \theta(x)\Big[ \int \{ 1\{d=1\} - p_0(x)\}\frac{g(y,d,x)}{f_0(y,d,x)} \frac{f_0(y, d,x)}{f_0(x)}dydd\Big] f_0(x)dx \\
	& = \int  \theta(x)\{ 1\{d=1\} - p_0(x)\}g(y,d,x)dydddx.
	\end{align*}
	Hence, 
	\begin{align*}
	&\frac{d}{d \tau}\Bigg[\int\Big( ((\eta_0(1,x)-\bar{y})p_\tau(x) + (\underline{y}-\eta_0(0,x))(1-p_\tau(x)))\theta_{10}(x)  \\
	&-((\eta_0(1,x)-\underline{y})p_\tau(x) + (\bar{y}-\eta_0(0, x))(1-p_\tau(x))) \theta_{01}(x))- \beta \Big)dF_{0}(x) \Bigg] \\
	& = \int  (\theta_{10}(x)-\theta_{01}(x))(\eta_0(1,x)+\eta_0(0,x) - (\underline y + \bar y)) \{ 1\{d=1\} - p_0(x)\}g(y,d,x)dydddx
	\end{align*}
	Therefore, as $j\to\infty$, since $\eta(1,x), \eta(0,x),$ and $p(x)$ are continuous at $x$, we obtain 
	\begin{align*}
	\phi(w, \beta, \gamma)  & = \Big((\eta_0(1,x)-\bar{y})p_0(x) + (\underline{y}-\eta_0(0,x))(1-p_0(x)))\theta_{10}(x)  \\
	&-((\eta_0(1,x)-\underline{y})p_0(x) + (\bar{y}-\eta_0(0, x))(1-p_0(x))) \theta_{01}(x)\Big)-\beta\\
	&+ (\theta_{10}(x)-\theta_{01}(x))(\eta_0(1,x)+\eta_0(0,x) - (\underline y + \bar y)) \{ 1\{d=1\} - p_0(x)\} \\
	& + (\theta_{10}(x)-\theta_{01}(x)) \{y - \eta_0(1,x)\}^d\{-(y - \eta_0(0,x))\}^{1-d}. \qed
	\end{align*} 

\subsection*{Proof of Theorem \ref{theorem:asymptoticvariance}}
Let 
\begin{equation}
\hat \psi(\beta_l) = \frac{1}{n} \sum_{k=1}^{L}\sum_{i \in \mathcal I_k} \psi(w_i, \beta_l, \hat{\gamma}_{k}).
\end{equation}
First we show that
\begin{equation} \label{equation:normalityproof}
\sqrt n \hat \psi (\beta_0) = \frac{1}{\sqrt n} \sum_{i=1}^{n} \psi(w_i, \beta_0, \gamma_0) + o_p(1)
\end{equation}
holds. Under Assumption \ref{assumption:ceinr_ass1} $(i)$, $(ii)$, and $(iii)$, the result follows. Following CEINR, we provide a sketch of the argument. Let 
\begin{equation}
\hat{\Delta}_{ik} \equiv \psi(w_i, \beta_0, \hat{\gamma}_{k})-\bar{\psi}(\hat{\gamma}_{k})-\psi(w_i, \beta_0, \gamma_0),
\end{equation}
and
\begin{equation}
\bar\Delta_{k} \equiv \frac{1}{n} \sum_{i \in \mathcal I_k}\bar\Delta_{ik}.
\end{equation}
Let $n_k$ be the number of observations with $i\in \mathcal I_k$ and $W_{k}$ denote a vector of all observations $w_i$ for $i\notin\mathcal I_k$. Note that for any $i,j \in \mathcal I_k$, $i\neq j$, we have $E[\hat\Delta_{ik}\hat\Delta_{jk}  | W_{k}]=E[\hat\Delta_{ik} | W_{k}]E[\hat\Delta_{jk} | W_{k}]=0$ since by construction 
$E[\hat\Delta_{ik} | W_{k}] = 0$.  By Assumption \ref{assumption:ceinr_ass1} $(i)$,
\begin{equation}
E[\bar\Delta_{k}^2| W_{k}] = \frac{1}{n^2} \sum _{i \in \mathcal I _k} E[\hat\Delta_{ik}^2| W_{k}] \leq \frac{n_k}{n^2} \int \{\psi(w,\beta_0, \hat\gamma_{k})-\psi(w, \beta_0, \gamma_0)\}^2F_0(dw) = o_p(n_k/n^2).
\end{equation}
This implies that, for each $k$, we have $\bar\Delta_{k}=o_p(\sqrt{n_k}/n)$. Then it follows that 
\begin{equation}
\sqrt n\Big[ \hat \psi (\beta_0) - \frac{1}{ n} \sum_{i=1}^{n} \psi(w_i, \beta_0, \gamma_0) - \frac{1}{n} \sum_{k=1}^{L} n_k \bar \psi(\hat\gamma_{k})\Big] = \sqrt n \sum_{k=1}^{L}\bar{\Delta}_k = o_p(\sqrt{n_k/n})\overset{p}{\longrightarrow} 0.
\end{equation}
By Assumption \ref{assumption:ceinr_ass1} $(ii)$ and $(iii)$, we have
\begin{equation}
\sqrt n |\bar \psi (\hat\gamma_{k}) | \leq \sqrt n  C \|\hat\gamma_{k}-\gamma_0\|^2 \overset{p}{\longrightarrow} 0.
\end{equation}
Then (\ref{equation:normalityproof}) follows by the triangle inequality. Since (\ref{equation:normalityproof}) holds and $\{w_i\}_{i=1}^{n}$ are i.i.d., by central limit theorem
\begin{equation}
\sqrt n \hat \psi (\beta_0)\overset{d}{\longrightarrow} N(0, \Omega),
\end{equation}
where $\Omega=E[\psi(w_i,\beta_0, \gamma_0)^2]$. The rest of the proof is standard as in \cite{newey1994large} and we provide a sketch of the argument.
Let $M = E[\frac{\partial \psi(w,\beta, \gamma_0)}{\partial \beta}|_{\beta=\beta_0}]$ and $\hat M = \frac{\partial \hat \psi(\hat \beta)}{\partial \beta}$. The first order condition is 
\begin{equation}
0=\hat M \hat \psi(\hat \beta).
\end{equation}
We expand $\hat \psi(\hat \beta)$ around $\beta_0$ to obtain 
\begin{equation}
\hat \psi(\hat \beta) = \hat \psi(\beta_0) + \bar M (\hat\beta-\beta_0),
\end{equation}
where $\bar M = \frac{\partial \hat \psi(\beta_u)}{\partial \beta}$ and $\bar{\beta}$ is the mean value.
Substituting this back into the first order condition, we get 
\begin{equation}
0= \hat M \hat \psi(\beta_0) + \hat M \bar M (\hat\beta-\beta_0).
\end{equation}
Solving this for $\hat{\beta}-\beta_0$ and multiplying by $\sqrt n$, we obtain
\begin{equation}
\sqrt n (\hat\beta-\beta_0) = -(\hat M\bar M)^{-1} \hat M \sqrt n \hat \psi(\beta_0).
\end{equation}
We also have $\hat M \overset{p}{\longrightarrow} M$ and $\bar M \overset{p}{\longrightarrow} M$ and by the continuous mapping theorem, 
\begin{equation}
-(\hat M\bar M)^{-1} \hat M \overset{p}{\longrightarrow} -M^{-1}.
\end{equation}
Then, by the Slutzky theorem, 
\begin{equation}
\sqrt n (\hat\beta-\beta_0) \overset{d}{\longrightarrow} -M^{-1}N(0,\Omega) = N(0,M^{-2}\Omega).
\end{equation}
In our case, $M=E[\frac{\partial \psi(w,\beta, \gamma_0)}{\partial \beta}|_{\beta=\beta_0}] = 1$ and so the asymptotic variance is $\Omega$.
Finally, CEINR showed that Assumption  \ref{assumption:ceinr_ass1} $(iii)$ insures that $\hat V \overset{p}{\longrightarrow} V$. \placeqed

\subsection*{Proof of Lemma \ref{lemma:influencefunction-iv-wc}}
The proof is similar to that of Lemma \ref{lemma:influencefunction}. Since we have an additional variable $Z\in\{0,1\}$, we make slight adjustments. 
	By the chain rule, we have
	\begin{align*}
	& \frac{d }{d\tau} E_\tau[m(w_i,\beta,\gamma(F_\tau))]
	= \frac{d }{d\tau} E_\tau[m(w_i,\beta,\gamma(F_0))] + \frac{d }{d\tau} E_0[m(w_i,\beta,\gamma(F_\tau))] \\
	& =  \frac{d}{d \tau} \Bigg[ \int\Big( \big(\eta_0(1, x, 1)\cdot p_0(x, 1) + \underline{y}\cdot(1-p_0(x, 1))- \bar{y}\cdot p_0(x,0) -\eta_0(0, x,0)\cdot(1-p_0(x,0))\big) \cdot \theta_{10}(x) \\
	& - \big( \eta_0(1, x, 1)\cdot p_0(x, 1) + \bar{y}\cdot(1-p_0(x,1))
	- \underline{y}\cdot p_0(x,0) - \eta_0(0, x,0)\cdot(1-p_0(x,0))\big) \cdot \theta_{01}(x) - \beta \Big)dF_{\tau}(x)   \Bigg]\\
	&+ \frac{d}{d \tau}\Bigg[\int\Big( \big(\eta_0(1, x, 1)\cdot p_\tau(x, 1) + \underline{y}\cdot(1-p_\tau(x, 1))- \bar{y}\cdot p_\tau(x,0) -\eta_0(0, x,0)\cdot(1-p_\tau(x,0))\big) \cdot \theta_{10}(x) \\
	& - \big( \eta_0(1, x, 1)\cdot p_\tau(x, 1) + \bar{y}\cdot(1-p_\tau(x,1))
	- \underline{y}\cdot p_\tau(x,0) - \eta_0(0, x,0)\cdot(1-p_\tau(x,0))\big) \cdot \theta_{01}(x) - \beta \Big)dF_{0}(x) \Bigg]\\
	&+ \frac{d}{d \tau} \Bigg[\int\Big( \big(\eta_\tau(1, x, 1)\cdot p_0(x, 1) + \underline{y}\cdot(1-p_0(x, 1))- \bar{y}\cdot p_0(x,0) -\eta_\tau(0, x,0)\cdot(1-p_0(x,0))\big) \cdot \theta_{10}(x) \\
	& - \big( \eta_\tau(1, x, 1)\cdot p_0(x, 1) + \bar{y}\cdot(1-p_0(x,1))
	- \underline{y}\cdot p_0(x,0) - \eta_\tau(0, x,0)\cdot(1-p_0(x,0))\big) \cdot \theta_{01}(x) - \beta \Big)dF_{0}(x)\Bigg].
	\end{align*}
	By the arguments in the proof of Lemma \ref{lemma:influencefunction}, we have 
	\begin{align}
	\frac{d}{d\tau}\int\theta(x) \eta_\tau(d,x,z)f_0(x)dx  
	= \int \theta(x) \{y - \eta_0(d,x,z)\}g(y,d,x)dydx.
	\end{align}
	and 
	\begin{align}
	\frac{d}{d\tau}\int \theta(x)p_\tau(x, z)f_0(x)dx 
	=  \int  \theta(x)\{ 1\{d=1\} - p_0(x,z)\}g(y,d,x)dydddx.
	\end{align}
	Hence, we obtain 
	\begin{align*}
	&\frac{d}{d \tau}\Bigg[\int\Big( \big(\eta_0(1, x, 1)\cdot p_\tau(x, 1) + \underline{y}\cdot(1-p_\tau(x, 1))- \bar{y}\cdot p_\tau(x,0) -\eta_0(0, x,0)\cdot(1-p_\tau(x,0))\big) \cdot \theta_{10}(x) \\
	& - \big( \eta_0(1, x, 1)\cdot p_\tau(x, 1) + \bar{y}\cdot(1-p_\tau(x,1))
	- \underline{y}\cdot p_\tau(x,0) - \eta_0(0, x,0)\cdot(1-p_\tau(x,0))\big) \cdot \theta_{01}(x) - \beta \Big)dF_{0}(x) \Bigg]\\
	& = \int ((\eta_0(1, x, 1) -\underline{y})\cdot \theta_{10}(x) - ( \eta_0(1, x, 1)- \bar y) \cdot \theta_{01}(x)) \cdot (d-p_0(x,1))  \\
	& \quad + ((\eta_0(0, x, 0) -\bar{y})\cdot \theta_{10}(x) - ( \eta_0(0, x, 0)- \underline y) \cdot \theta_{01}(x)) \cdot (d-p_0(x,0)) g(y,d,z,x)dydddzdx
	\end{align*}
	Also, 
	\begin{align*}
	& \frac{d}{d \tau} \Bigg[\int\Big( \big(\eta_\tau(1, x, 1)\cdot p_0(x, 1) + \underline{y}\cdot(1-p_0(x, 1))- \bar{y}\cdot p_0(x,0) -\eta_\tau(0, x,0)\cdot(1-p_0(x,0))\big) \cdot \theta_{10}(x) \\
	& - \big( \eta_\tau(1, x, 1)\cdot p_0(x, 1) + \bar{y}\cdot(1-p_0(x,1))
	- \underline{y}\cdot p_0(x,0) - \eta_\tau(0, x,0)\cdot(1-p_0(x,0))\big) \cdot \theta_{01}(x) - \beta \Big)dF_{0}(x)\Bigg] \\
	& = \int (\theta_{10}(x)-\theta_{01}(x))\cdot p_0(x,1) \cdot (y-\eta_0(1, x, 1)) g(y,d,z,x)dydddzdx\\
	& \quad - \int (\theta_{10}(x)-\theta_{01}(x))\cdot (1-p_0(x,0)) \cdot (y-\eta_0(0, x, 0)) g(y,d,z,x)dydddzdx
	\end{align*}\placeqed

\section{More General Case} \label{appendix:moregeneralcase}
I show how my result can be extended to a more general setting. Let $\delta: \mathcal{X} \rightarrow [0,1]$ so that the treatment rules can be randomized treatment rules. Also, let $w: \mathcal X \rightarrow \mathbb {R}_{+}$ be some weighting function so that the policymaker cares about the weighted average welfare $E[w(X)\cdot(E[Y_1|X]\cdot\delta(X) + E[Y_0|X]\cdot(1-\delta(X)))]$ rather than the mean welfare. Then, by letting $\psi(X)\equiv w(X)\cdot(\delta(X)-\delta^*(X))$, my object of interest becomes
\begin{align}
E[\Delta(X)\cdot w(X)\cdot(\delta(X)- \delta^*(X))] = E[\Delta(X)\cdot\psi(X)].
\end{align}
I derive the identification of this parameter in the following theorem. 
\begin{theorem}[More general case]
	Suppose $(\mathcal{Y}_1 \times \mathcal{Y}_0): \Omega \to \mathcal F$ is an integrable random set.  Let $\delta: \mathcal{X} \rightarrow [0,1]$ and $\delta^*: \mathcal{X} \rightarrow [0,1]$ be treatment rules and $w: \mathcal X \rightarrow \mathbb {R}_{+}$ be a weighting function. Also, let $\psi(X) \equiv w(X)\cdot(\delta(X)-\delta^*(X))$ and $p=(1,-1)'$. Then, $B_I(\delta, \delta^*)$ in \eqref{equation:id} is an interval $[\beta_l, \beta_u]$ where
	\begin{equation} \label{lower}
	\beta_l = E[\underline{\Delta}(X)\cdot|\psi(X)|\cdot \mathbbm{1}\{\psi(X)\geq 0\} - \bar{\Delta}(X)\cdot|\psi(X)|\cdot \mathbbm{1}\{\psi(X) < 0\}],
	\end{equation}
	and
	\begin{equation} \label{upper}
	\beta_u = E[\bar{\Delta}(X)\cdot|\psi(X)| \cdot \mathbbm{1}\{\psi(X)\geq 0\} -\underline{\Delta}(X)\cdot|\psi(X)|\cdot \mathbbm{1}\{\psi(X) < 0\}],
	\end{equation}
	where $\bar{\Delta}(X) \equiv E[s(v^*,\mathcal{Y}_1 \times \mathcal{Y}_0)|X]$ and $\underline{\Delta}(X)\equiv - E[s(-v^*,\mathcal{Y}_1 \times \mathcal{Y}_0)|X]$.
\end{theorem}

\bigskip

\begin{proof}
	The proof is similar to that of Theorem \ref{theorem:main}. I still have (\ref{equation:delta}) to bound $\Delta(X)$. Since $\psi(X)\in \mathbb R$, I consider two cases: (i) $\psi(X)\geq 0$ and  (ii) $\psi(X)<0$. When (i) $\psi(X)\geq 0$, the upper bound on $\Delta(X)\cdot\psi(X)$ is $\bar\Delta(X)\cdot|\psi(X)|$. When (ii) $\psi(X)<0$, the upper bound on $\Delta(X)\cdot\psi(X)$ is $-\underline\Delta(X)\cdot|\psi(X)|$. Hence,  the lower bound should be 
	\begin{equation}
	\beta_l = E[\underline{\Delta}(X)\cdot|\psi(X)|\cdot \mathbbm{1}\{\psi(X)\geq 0\} - \bar{\Delta}(X)\cdot|\psi(X)|\cdot \mathbbm{1}\{\psi(X) < 0\}]. 
	\end{equation}
	Similarly, the upper bound on $E[\Delta(X)\cdot\psi(X)]$ should be 
	\begin{equation}
	\beta_u = E[\bar{\Delta}(X) \cdot|\psi(X)|\cdot \mathbbm{1}\{\psi(X)\geq 0\} -\underline{\Delta}(X)\cdot|\psi(X)|\cdot \mathbbm{1}\{\psi(X) < 0\}]. \qedhere
	\end{equation}
\end{proof}

\clearpage
\section{More Details on the Simulation Study}\label{appendix:simulation}
The population welfare gain that results from switching from $\delta^*(x)=1\{x \leq11\}$ to $\delta(x)=1\{x \leq12\}$ is
\begin{align}
\begin{split}
\beta & = E\big[E[Y_1-Y_0|X]\cdot(\delta(X)- \delta^*(X))\big] \\
& = P(X=12)\cdot E[Y_1-Y_0|X=12] \\
& = 0.43 \cdot \int_{0}^{1} \{ m_1(12,u) - m_0(12,u)\} du \\
& = 1236.
\end{split} 
\end{align}
The integration is done using integrate() function on R.
Given the structure in Section \ref{section:simulation}, we have
\begin{align}
E[Y|D=0,X,Z] &  =E[Y_0|U>p(X,Z), X,Z] = \frac{1}{1-p(X,Z)} \int_{p(X,Z)}^{1} m_0(X,u) du,\\
E[Y|D=1,X,Z] & =E[Y_1|U \leq p(X,Z), X,Z] = \frac{1}{p(X,Z)} \int_{0}^{p(X,Z)} m_1(X,u) du,\\
P(D=1|X) & = 2/3\cdot p(X,1) + 1/3\cdot p(X,0),\\
P(Z=1|D=0,X) & = \frac{P(D=0|Z=1,X)P(Z=1|X)}{P(D=0|X)} = \frac{2/3\cdot(1-p(X,1))}{1-2/3\cdot p(X,1)-1/3\cdot p(X,0)}, \\
P(Z=0|D=0,X) & = \frac{1/3\cdot (1-p(X,0))}{1-2/3\cdot p(X,1)-1/3\cdot p(X,0)}, \\
P(Z=1|D=1,X) & = \frac{P(D=1|Z=1,X)P(Z=1|X)}{P(D=1|X)} = \frac{2/3\cdot p(X,1)}{2/3\cdot p(X,1)+1/3\cdot p(X,0)}, \\
P(Z=0|D=1,X) & = \frac{1/3\cdot p(X,0)}{2/3\cdot p(X,1)+1/3\cdot p(X,0)},
\end{align}
\begin{align}
\begin{split}
E[Y|D=0,X] & = E[Y|D=0,X,Z=1]\cdot P(Z=1|D=0,X) \\
& \quad + E[Y|D=0,X,Z=0]\cdot P(Z=0|D=0,X), \\
\end{split}
\end{align}
\begin{align}
\begin{split}
E[Y|D=1,X] & = E[Y|D=1,X,Z=1]\cdot P(Z=1|D=1,X) \\
& \quad + E[Y|D=1,X,Z=0]\cdot P(Z=0|D=1,X).
\end{split}
\end{align}
Given these quantities, worst-case and IV-worst-case bounds can be calculated similarly. 

\clearpage
\section{Additional Tables}
\begin{table}[!htbp] \centering 
	\caption{Empirical means and distributions when $X$ is years of education} 
	\label{table:emp_education} 
	\resizebox{\textwidth}{!}{
	\begin{threeparttable} 
	\begin{tabular}{@{\extracolsep{5pt}} cccccccccccccc} 
		\hline \\[-1.8ex] 
		$X$ & 7 & 8 & 9 & 10 & 11 & 12 & 13 & 14 & 15 & 16 & 17 & 18 & Total \\
		\hline 
		& \\
		sample size & 34 & 616 & 642 & 984 & 1167 & 3940 & 660 & 602 & 197 & 260 & 111 & 10 & 9223 \\ 
		&\\
		$P(X)$ & 0.004 & 0.067 & 0.07 & 0.107 & 0.127 & 0.427 & 0.072 & 0.065 & 0.021 & 0.028 & 0.012 & 0.001 & 1 \\ 
		&\\
		$E[Y|X]$ & 7998 & 12252 & 12509 & 14095 & 13492 & 16982 & 18210 & 20204 & 20837 & 20875 & 20032 & 11606 \\ 
		&\\
		$E[Y|D=1,X]$ & 7747 & 14916 & 14860 & 14706 & 15622 & 18391 & 18713 & 21093 & 21369 & 20678 & 22082 & 9207 \\ 
		&\\
		$E[Y|D=0,X]$ & 8102 & 10469 & 10908 & 13662 & 12014 & 15786 & 17745 & 19520 & 20322 & 21033 & 18411 & 14005 \\ 
		&\\
		$P(D=1|X)$ & 0.294 & 0.401 & 0.405 & 0.415 & 0.41 & 0.459 & 0.48 & 0.435 & 0.492 & 0.446 & 0.441 & 0.5 \\ 
		&\\
		$E[Y|D=1,Z=1,X]$ & 7747 & 14932 & 14860 & 14687 & 15639 & 18448 & 18833 & 21272 & 21369 & 20678 & 22082 & 9207 \\ 
		&\\
		$E[Y|D=1,Z=0, X]$ & 0 & 11011 & 14886 & 17263 & 14000 & 13530 & 6089 & 13449 & 0 & 0 & 0 & 0 \\ 
		&\\
		$E[Y|D=0,Z=1, X]$ & 6194 & 9778 & 10114 & 13944 & 11752 & 15345 & 16299 & 18776 & 21010 & 16767 & 16904 & 12636 \\ 
		&\\
		$E[Y|D=0,Z=0, X]$ & 8888 & 10980 & 11546 & 13451 & 12196 & 16084 & 18661 & 20028 & 19781 & 23669 & 19429 & 14918 \\ 
		&\\
		$P(D=1|Z=1, X)$ & 0.588 & 0.61 & 0.601 & 0.622 & 0.626 & 0.676 & 0.702 & 0.65 & 0.688 & 0.678 & 0.662 & 0.714 \\ 
		&\\
		$P(D=1|Z=0, X)$ & 0 & 0.005 & 0.019 & 0.009 & 0.012 & 0.016 & 0.014 & 0.029 & 0 & 0 & 0 & 0 \\ 
		&\\
		\hline \\[-1.8ex] 
	\end{tabular}
	\begin{tablenotes}
	\item This table reports the empirical means and distributions when $X$ is years of education. Y denotes the outcome variable which is 30-month earnings in US Dollars. $D$ denotes the program participation and equals 1 for individuals who participated in the program. $Z$ denotes the random assignment to treatment and equals 1 for individuals who got offered job training.
\end{tablenotes}
\end{threeparttable} }
\end{table}

\clearpage

\section{Additional Figures}

\begin{figure}[h]
	\centering
	\caption{Estimated conditional mean treatment responses\\ when X is years of education}
	
	\includegraphics[width=0.4\linewidth]{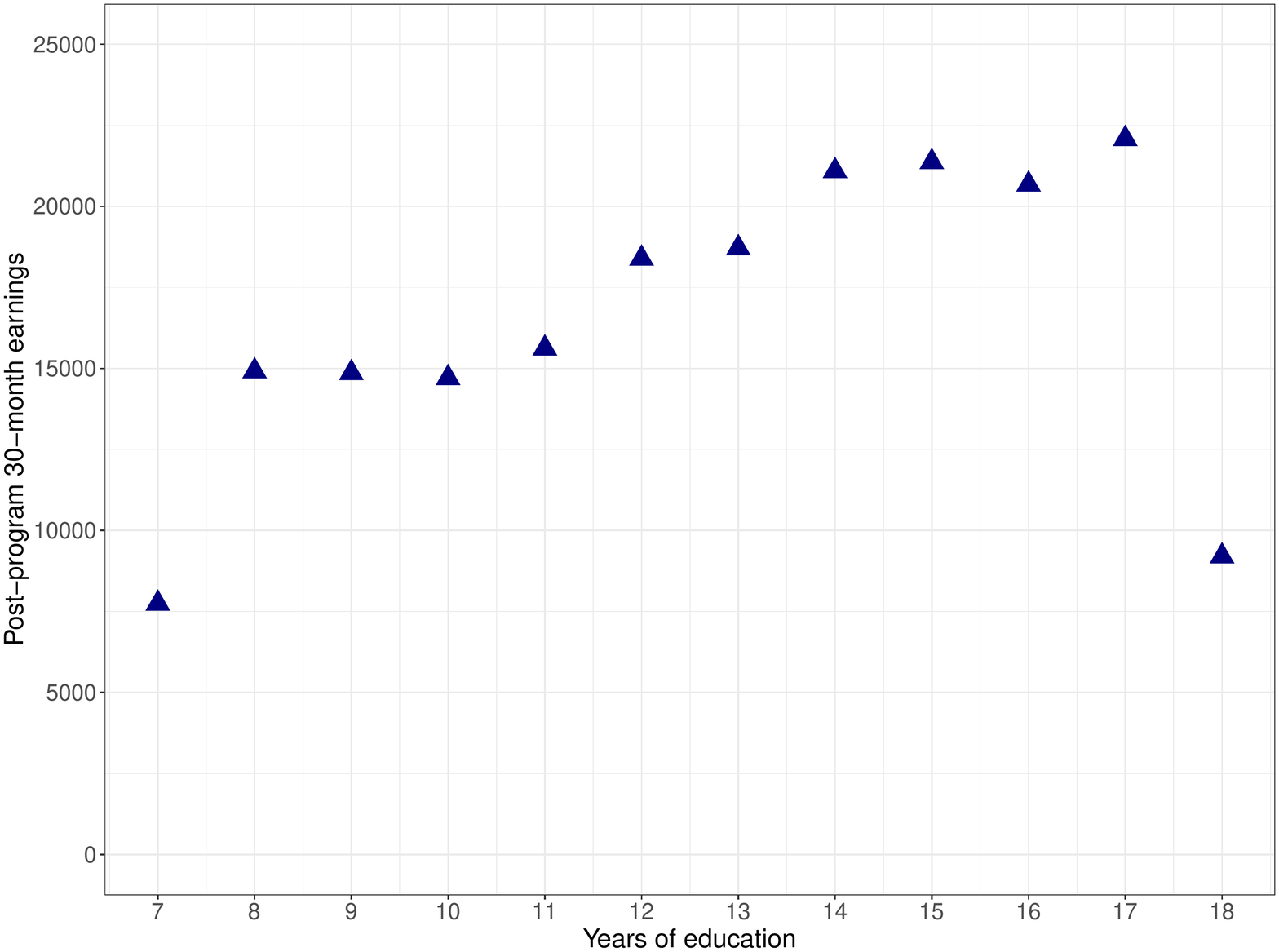}
	\includegraphics[width=0.4\linewidth]{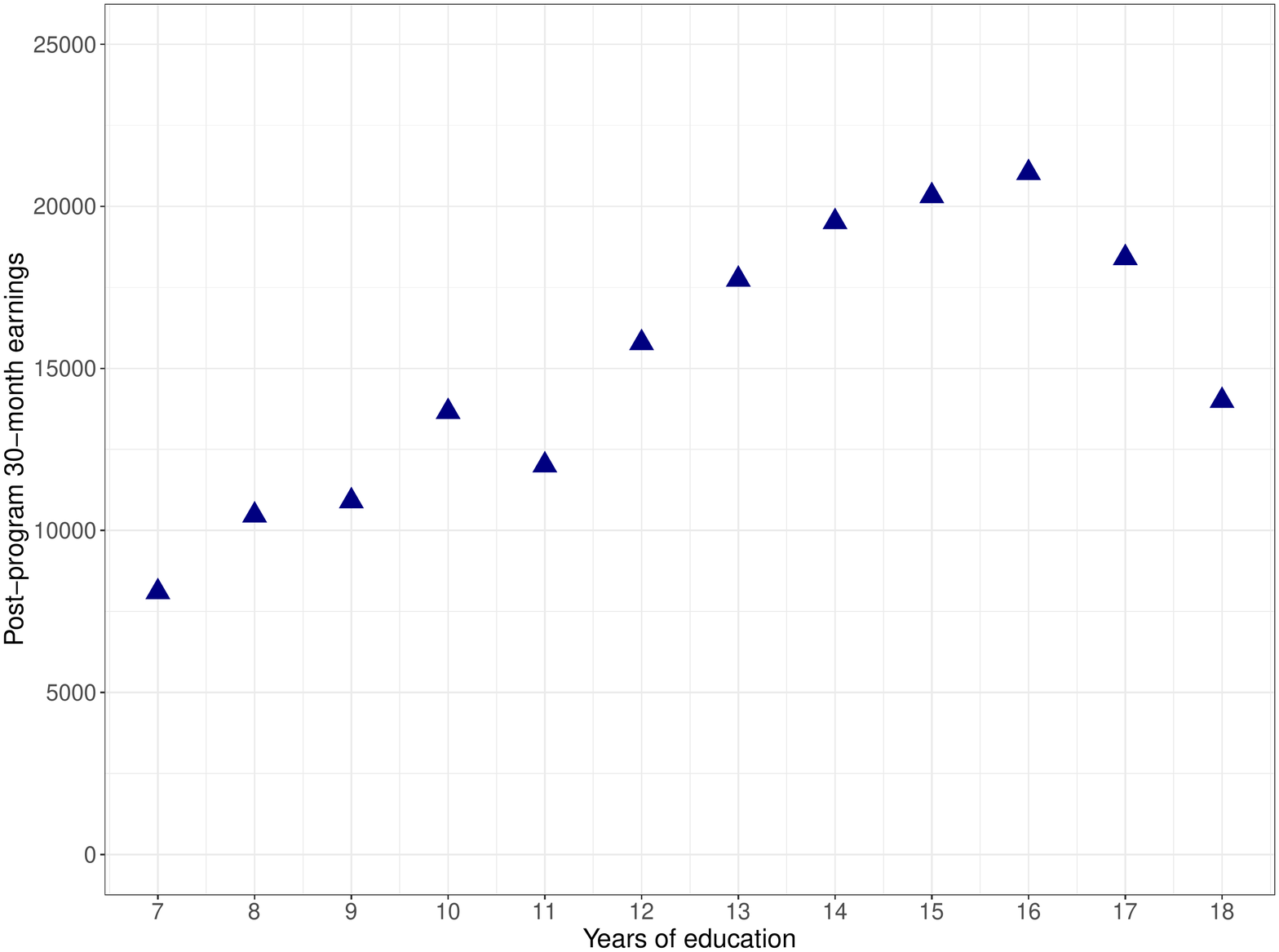}
	\label{fig:eta-edu}
	
	\centering 	\small Left figure is for $\hat\eta(1,X)$ and right figure is for $\hat\eta(0,X)$
		 
\end{figure}

\bigskip

\begin{figure}[h]	
	\centering
	\caption{Estimated propensity score \\ when X is years of education}
	\includegraphics[width=0.4\linewidth]{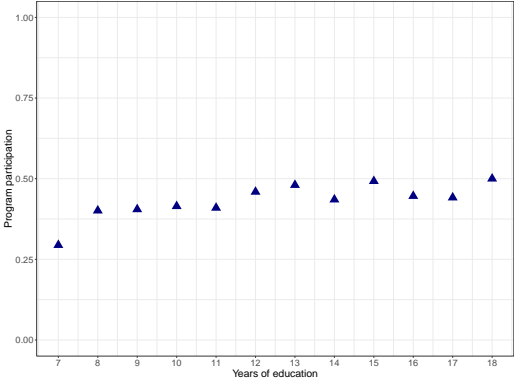}

	\label{fig:pedu}
\end{figure}

\begin{figure}[h]
	\centering
	\caption{Hypothetical policies considered in Example 2 \\ in Empirical Application}
	\includegraphics[width=0.7\linewidth]{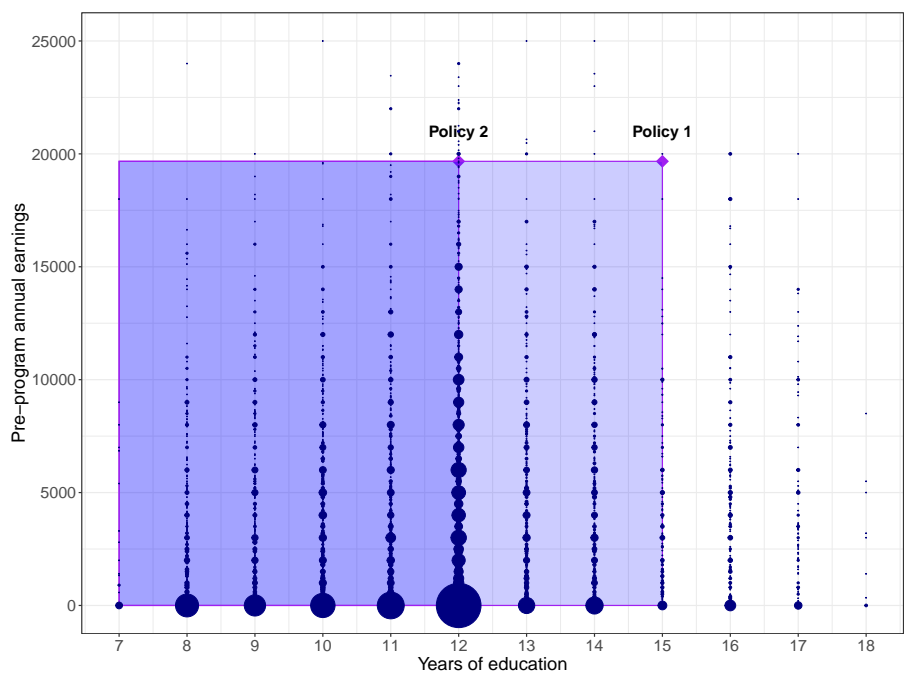}
	\label{fig:policies}
	\bigskip
	
	\centering \small
	Policy 1 is $\mathbbm1\{\text{education} \leq15, \text{pre-program earnings} \leq \$19,670\},$ \\ and Policy 2 is $\mathbbm1\{\text{education} \leq12, \text{pre-program earnings} \leq \$19,670\}.$
\end{figure}

\clearpage

\begin{figure}[h]
	\centering
	\caption{Estimated conditional mean treatment response \\ when X is years of education and pre-program annual earnings}
	\includegraphics[width=0.4\linewidth]{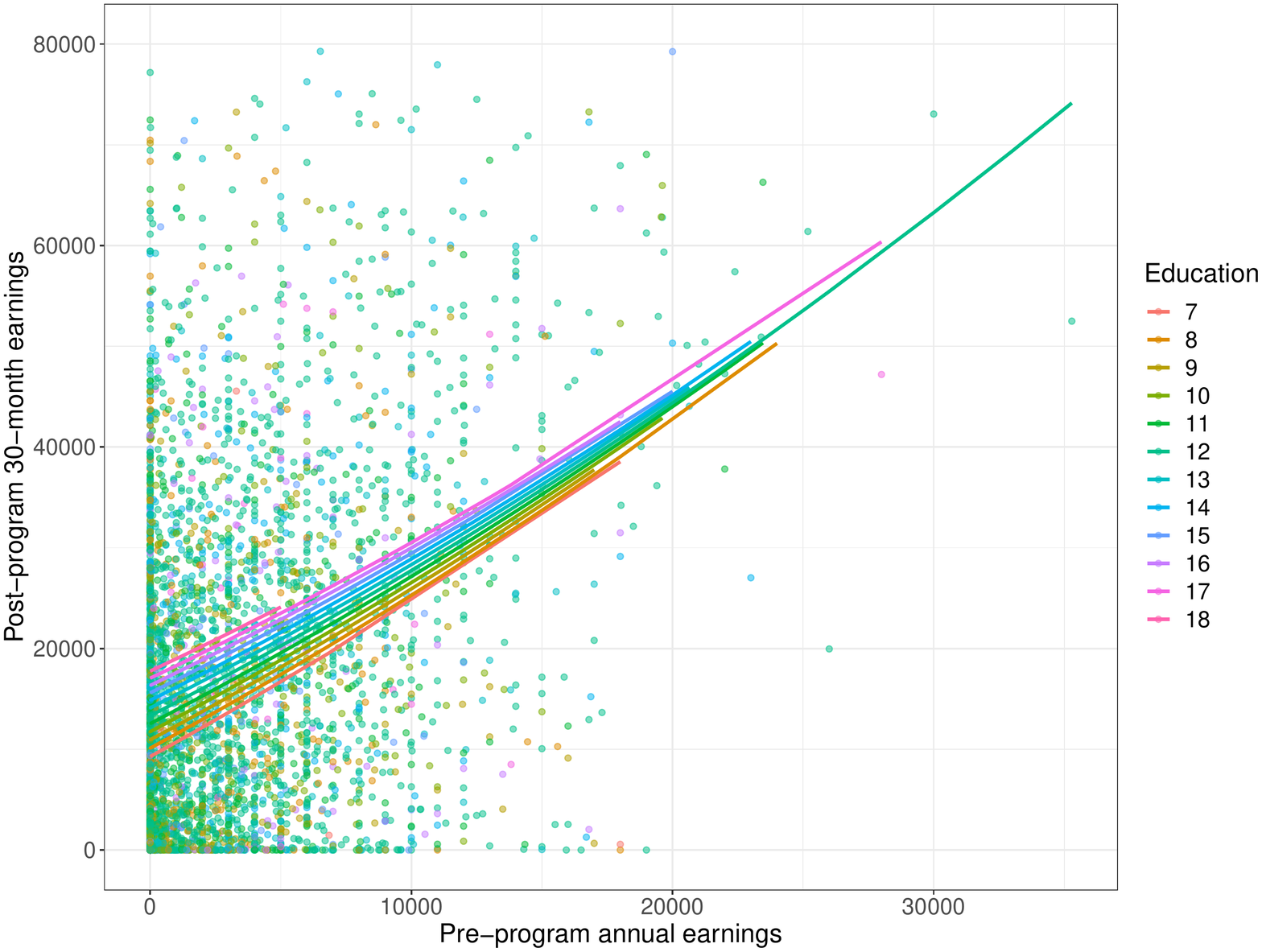}
		\includegraphics[width=0.4\linewidth]{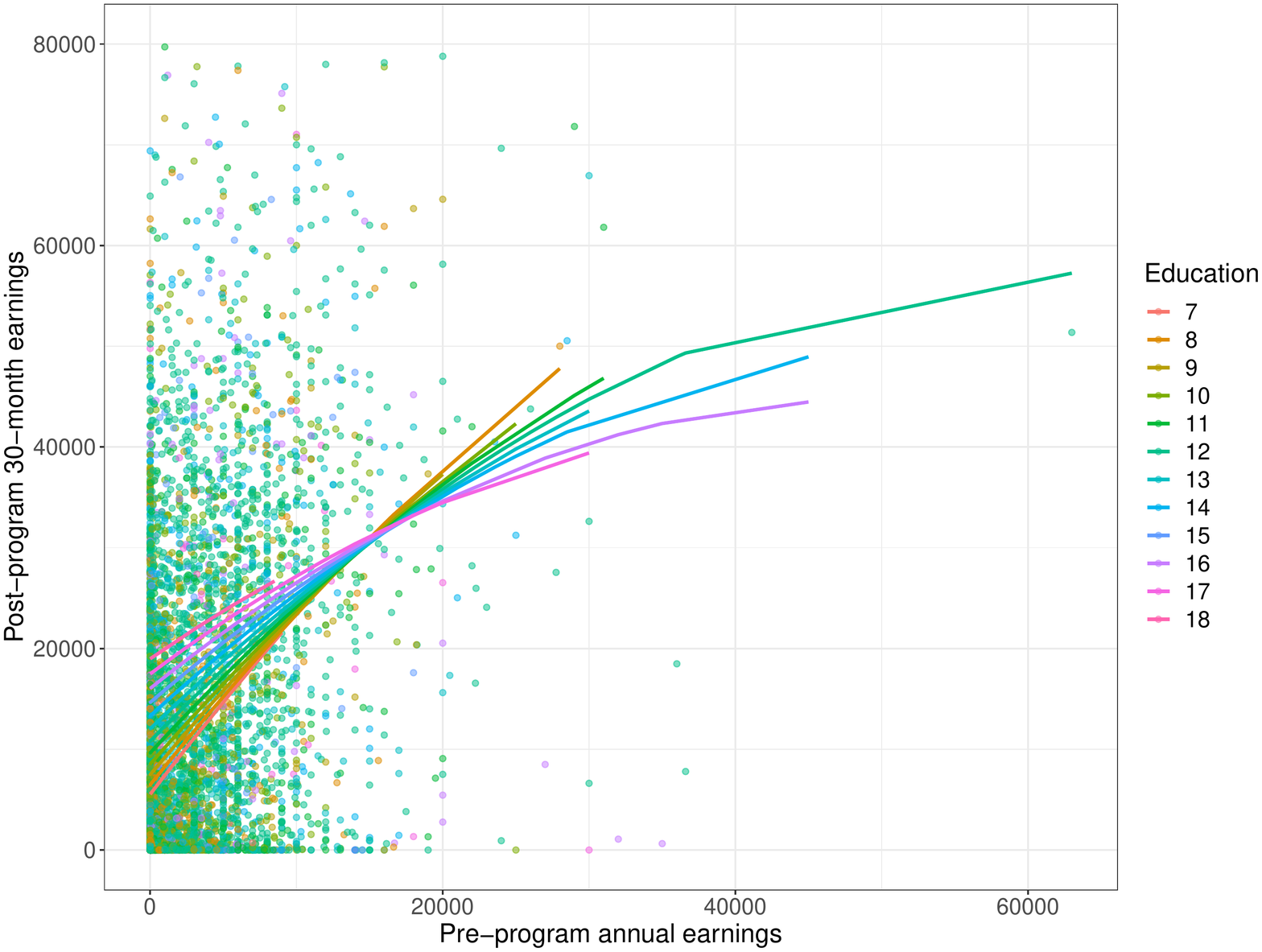}
	\label{fig:etaeduprevearn}
	
	\centering 	\small Left figure is for $\hat\eta(1,X)$ and right figure is for $\hat\eta(0,X)$
\end{figure}

\begin{figure}[h]
	\centering
	\caption{Estimated propensity score \\ when X is years of education and pre-program annual earnings}
	\includegraphics[width=0.4\linewidth]{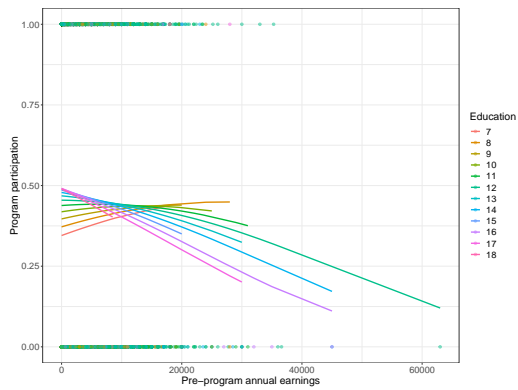}

	\label{fig:peduprevearn}
\end{figure}

\clearpage
\singlespacing
\bibliography{welfare}

\end{document}